\pdfoutput=1
\documentclass[twoside,11pt]{Classes/CUEDthesisPSnPDF}
\usepackage{amsmath}
\usepackage{amsfonts}
\usepackage{amssymb}
\usepackage{graphicx}

\newtheorem{thm}{Theorem}
\newtheorem{lem}{Lemma}
\newtheorem{col}{Corollary}
\newtheorem{proof}{Proof}

\title{\textsc{Entanglement-Enhanced Classical Communication}}

  \author{David A. Herrera Mart\'i}
  \manager{Dr. Baltasar Beferull-Lozano}
  \managertwo{}
  \collegeordept{Department of Computer Science \\ School of Engineering}
  \crest{\includegraphics[width=40mm]{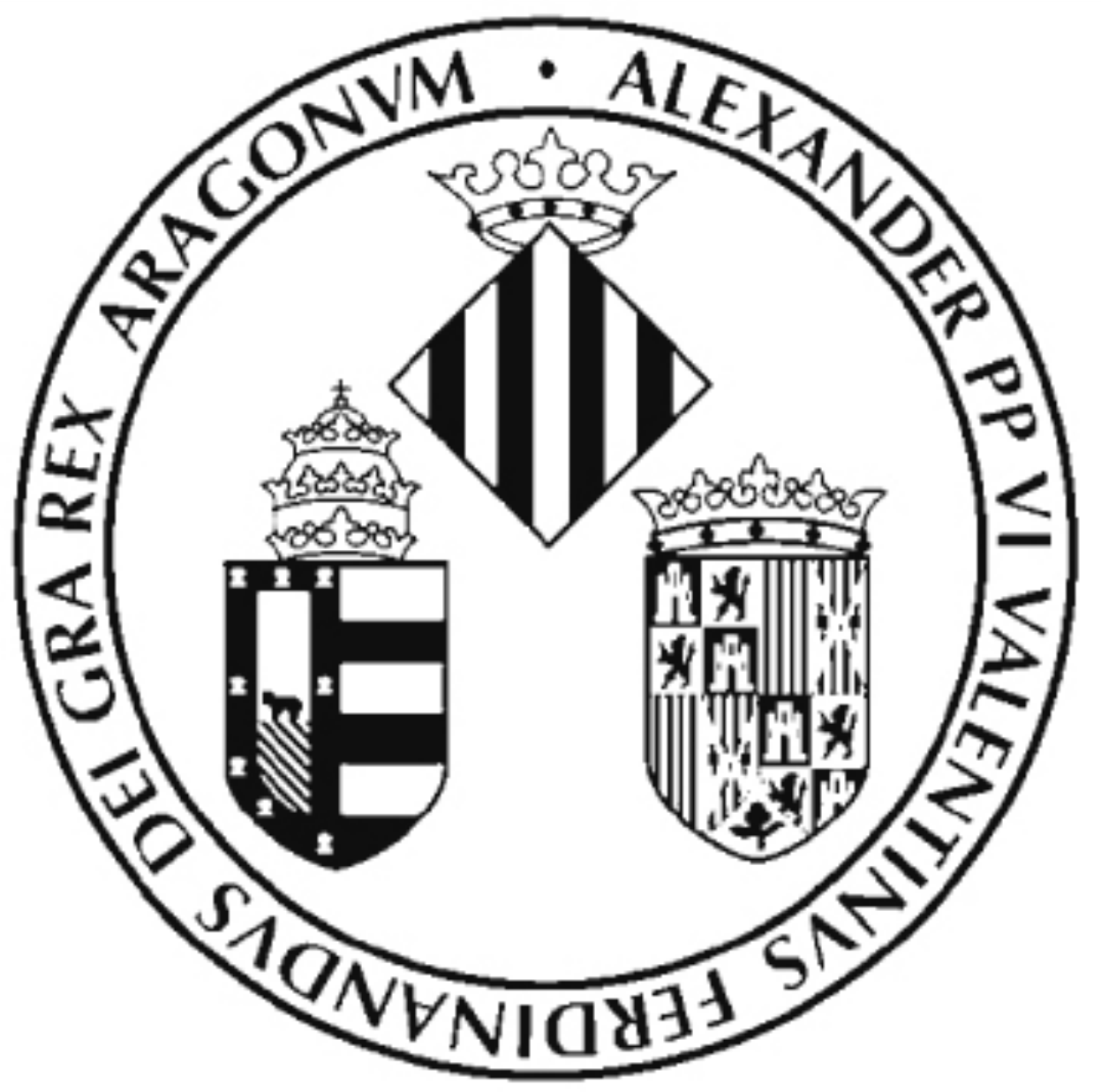}}
  \degree{Thesis submitted for the Degree of Master of Science (M.Sc.)}
  \university{University of Valencia}
  \degreedate{August 2008}

\renewcommand{\baselinestretch}{1}
\begin{document}
\maketitle
\graphicspath{{figuras/}}

\setcounter{secnumdepth}{3} \setcounter{tocdepth}{3}

\renewcommand{\baselinestretch}{1.2}

\begin{romanpages}
\begin{flushright}
\textsl{A mi madre.}

\textsl{En memoria de mi padre.}
\end{flushright}
\cleardoublepage

\tableofcontents 

\chapter*{Foreword}

This work was born from my desire to unify my two scientific backgrounds: physics and telecommunications. The study of Quantum Information Theory has given me the opportunity to complement these two subjects. The quantum theory predicts significant changes in our concept of computation and information. The conceptual jump from mathematical models to physical reality has outstanding consequences, such as new paradigm of complexity classes, in the case of computation, which allows for solving problems believed to be in NP, such as the \emph{Factoring Problem} or \emph{Discrete Logarithm Problem}, in polynomial time.

This thesis will be focused on the classical capacity of quantum channels, one of the first areas treated by quantum information theorists. The problem is fairly solved since some years. Nevertheless, this work will give me a reason to introduce a consistent formalism of the quantum theory, as well as to review fundamental facts about quantum non-locality and how it can be used to enhance communication. Moreover, this reflects my dwelling in the spirit of classical information theory, and it is intended to be a starting point towards a thorough study of how quantum technologies can help to shape the future of telecommunications.

Whenever it was possible, heuristic reasonings were introduced instead of rigorous mathematical proofs. This finds an explanation in that I am a self-taught neophyte in the field, and just about every time I came across a new concept, physical arguments were always more compelling to me than just maths. The technical content of the thesis is twofold. On one hand, a quadratic classification based on optimization programs that I devised for distinguishing entangled states is presented in Chapter 4. In second place, a less difficult yet I hope equally interesting technical part consists of versions of some proofs throughout the text.

Thanks are due to Teresa and to my family for their encouragement and support. All my love is for them. Almost all figures were possible thanks to the skills of Pedro Jim\'{e}nez. I also acknowledge pleasant conversations and support from all components of the GSIC group in Valencia.

My final words are for my father, suddenly deceased in late may of this year. He taught me the need for curiosity and joy. I owe him most of myself.

\bigskip
\begin{flushright}
Valencia, August 2008
\end{flushright}
\end{romanpages}

\part{Brief Review of Information Theory and Quantum Mechanics}

\chapter{A Mathematical Model for Communication}

Information Theory is mainly concerned about two issues. First one is to establish theoretical bounds to the achievable rates at which information can be compressed from a \emph{source} and conveyed through a \emph{channel}.  To this goal, achievability and converse theorems for different communication scenarios must be found. However, it is important to realize that these theorems are regardless of the complexity and delay of the codes that should attain the bounds. In second place, Information Theory is aimed at finding \emph{practical coding schemes} that perform close to theoretical limits. In this dissertation we will study exclusively first one of these two problems, to which Quantum Mechanics has endowed with a even richer variety of problems.

Most of this chapter is based on the texts \cite{cover2006eit}\cite{csiszar1982itc}\cite{gray1990eai}. Since this chapter is a review of basic concepts, results about stochastic processes and typicality will not be proved.

\section{What is Information?}

Before starting maybe one should face the question ``\emph{what is information?}". How should this ubiquitous and quasi-philosophical process be described mathematically? It seems natural to define information in terms of probability theory, for it is the mathematical framework that formally incorporates the concept of uncertainty about the future.

In the first half of the past century, several meaningful definitions arose, such as Fisher's information (which is a measure of the curvature of the probability distribution) or Hartley's function (the logarithm of the source's alphabet size), in the context of statistics and engineering\footnote{More general and deeper concepts such as R\'enyi's entropy or Kolmogorov's algorithmic complexity and their far reaching implications are not discussed here for the sake of conciseness}.

In 1948, guided by some reasonable assumptions, Shannon came out with \emph{entropy}, H, as a measure for information.Among others, his requirements were that:

\begin{enumerate}
  \item $H(\mathbf{p})$ be continuous on $\mathbf{p}$ ($\mathbf{1}^{T}\mathbf{p} = 1$)
  \item $H(\mathbf{p})$ should be, for $p_{i} = \frac{1}{n}$ a monotonic increasing function of n. This is equivalent to a normalization.
  \item If a choice is broken down into successive choices, the original entropy should be the weighted sum of individual values of the resulting entropies.
\end{enumerate}

It can be shown that the only function, up to a proportionality constant, satisfying these assumptions is\footnote{Throughout this thesis, logarithms will be taken in base 2.}:

\begin{equation}
    H(p) = -\sum^{n}_{i = 1} p_{i} \log p_{i}
\end{equation}

In fact, Shannon's entropy is the epigone of deeper concepts such as the \emph{relative entropy} (also known as Kullback-Leibler distance), or \emph{mutual information}. The relative entropy of two probability distributions is given by:

\begin{equation}
D(\mathbf{p} || \mathbf{q}) = \sum^{n}_{i=1} p_{i} \log \frac{p_{i}}{q_{i}}
\end{equation}

with $\mathbf{1}^{T}\mathbf{p} = \mathbf{1}^{T}\mathbf{q} = 1$, that is, $\mathbf{p}$ and $\mathbf{q}$  belong to the discrete probability simplex of dimension n, $\mathfrak{P}_{n}$. Although in general $D(\mathbf{p} || \mathbf{q})\neq D(\mathbf{q} || \mathbf{p})$, this quantity can be thought of as ``distance" between probability distributions.

Mutual information is the amount of information that a random variable contains about another random variable\footnote{Note that throughout this text, we will often interchange random variables for their induced probability distributions, and viceversa}. Consider X taking values in $\mathcal{X}$ and Y taking values in $\mathcal{Y}$, and let $I(X;Y)$ denote $I(\mathbf{p}^{X};\mathbf{q}^{Y})$, then their mutual information is:

\begin{equation}
    I(X;Y) = \sum^{n}_{i=1}\sum^{m}_{j=1} P^{XY}_{i,j} \log \frac{P^{XY}_{i,j}}{p^{X}_{i}q^{Y}_{j}}
\end{equation}

where $\mathbf{P}^{XY}$ is the joint probability distribution of both random variables. If they are independent, mutual information vanishes, which means that knowing the realization of one random variable does not give any clue about the other one.

In turn, Shannon's entropy is a special case of mutual information, being the information that a random variable contains about itself, $H(\mathbf{p}^{X}) = I(X;X)$. It will suffice to prove some properties of the relative entropy, because they can be straightforwardly extended to mutual information and entropy.

\begin{thm}\label{thm:relent}\textcolor[rgb]{0.00,0.00,1.00}{[Nonnegativity of relative entropy]}The relative entropy is positive semidefinite, $D(\mathbf{p} || \mathbf{q}) \geq 0$
\end{thm}
\begin{proof} Let $A = Supp(\mathbf{p})$ be the support of $\mathbf{p}$. Then

\begin{eqnarray}
  D(\mathbf{p} || \mathbf{q}) &=& - \sum_{i\in A} p_{i} \log \frac{q_{i}}{p_{i}} \nonumber\\
   &\geq& - \log \sum_{i\in A} p_{i}\frac{q_{i}}{p_{i}} \nonumber\\
   &=& - \log \sum_{i\in A} q_{i}\nonumber
\end{eqnarray}

\begin{eqnarray}
   &\geq& - \log \sum^{n}_{i=1} q_{i}\nonumber\\
   &=& - \log 1 \nonumber\\
   &=& 0
\end{eqnarray}

the first inequality is a consequence of \emph{Jensen's inequality} for convex functions $E[f(\mathbf{p}^{X})] \geq f(E[\mathbf{p}^{X}])$. The second inequality comes from extending the range of the sum.
\end{proof}

\begin{thm}\textcolor[rgb]{0.00,0.00,1.00}{[Convexity of relative entropy]}The relative entropy is a convex function of the probability distributions $\mathbf{p}$ and $\mathbf{q}$
\end{thm}
\begin{proof} By the \emph{log sum inequality} $\sum^{n}_{i=1} a_{i} \log \frac{a_{i}}{b_{i}} \geq (\sum^{n}_{i=1} a_{i}) \log \frac{\sum^{n}_{i=1} a_{i}}{\sum^{n}_{i=1} b_{i}}$ \cite{cover2006eit}, we have that:

\begin{equation}
(\lambda p_{i} + (1-\lambda)p'_{i})\log\frac{\lambda p_{i} + (1-\lambda)p'_{i}}{\lambda q_{i} + (1-\lambda)q'_{i}}\leq\lambda p_{i}\log\frac{\lambda p_{i}}{\lambda p_{i}} + (1-\lambda)p'_{i}\log\frac{(1-\lambda)p'_{i}}{(1-\lambda)q'_{i}}
\end{equation}

with $\lambda \in [0,1]$. Summing over the index we get:

\begin{equation}
D(\mathbf{p} + (1-\lambda)\mathbf{p}'||\lambda \mathbf{q} + (1-\lambda)\mathbf{q}')\leq \lambda D(\mathbf{p}||\mathbf{q}) + (1 - \lambda)D(\mathbf{p}'||\mathbf{q}')
\end{equation}

\end{proof}

\begin{col}\textcolor[rgb]{0.00,0.00,1.00}{[Concavity of entropy]}Entropy is a concave function of $\mathbf{p}^{X}$
\end{col}
\begin{proof} Consider the uniform distribution $\mathbf{u}^{X}=\frac{1}{\|\mathcal{X}\|}(1,1,..., 1)$. The relative entropy of distribution $\mathbf{p}^{X}$ with respect to $\mathbf{u}^{X}$ is:

$$D(p||u) = \sum^{n}_{i=1}p_{i}\log p_{i} - \sum^{n}_{i=1}p_{i}\log u_{i} = - H(\mathbf{p}^{X}) + \log \|\mathcal{X}\|$$

so we get:

\begin{equation}
H(\mathbf{p}^{X}) = \log \|\mathcal{X}\| - D(p||u)
\end{equation}

\end{proof}

It is easy to see from this corollary that entropy is upper bounded by the logarithm of the cardinality of the alphabet $H(p^{X}) \leq \log \|\mathcal{X}\| $.

A related important quantity is the \emph{conditional entropy} of a random variable Y given that the instantiation of X is known, i.e. the residual uncertainty about Y once we learn about X.

\begin{equation}
    H(Y|X=x_{i}) = - \sum^{m}_{j=1} q^{Y|X}_{j} \log q^{Y|X}_{j}
\end{equation}

Averaging over all possible outcomes of X:

\begin{equation}
    H(Y|X) = \sum^{n}_{i=1} p^{X}_{i} H(Y | X = x_{i}) = - \sum^{n}_{i=1} \sum^{m}_{j=1} P^{XY}_{i,j} \log q^{Y|X}_{j}
\end{equation}

Clearly there is a reduction in the uncertainty only if there exist a non-factorizable joint probability distribution. In other words, if the two random variables are independent, then $H(Y|X) = H(Y)$. By symmetry arguments one can easily find the relations:

\begin{equation}
    H(X) - H(X|Y) = I(X;Y) = H(Y) - H(Y|X)
\end{equation}

\begin{equation}
    I(X;Y)\leq\min\{H(X), H(Y)\}
\end{equation}

One useful property which makes use of the conditional entropy is the \emph{chain rule for entropy}. Let X, Y and Z be three random variables, then their joint entropy can be written:

\begin{equation}
    H(X,Y,Z) = H(X) + H(Y|X) + H(Z|X,Y)
\end{equation}

which is easily generalizable to any number of random variables.

The convexity of relative entropy has an important consequence for channel coding, as we will see:

\begin{thm}\textcolor[rgb]{0.00,0.00,1.00}{[Partial concavity of mutual information]}For fixed $\mathbf{p}^{Y|X}$, the mutual information is a concave function of $\mathbf{p}^{X}$
\end{thm}
\begin{proof} From Bayes' rule:

\begin{equation}
\mathbf{q}^{Y} = \frac{\mathbf{p}^{Y|X}}{\mathbf{q}^{X|Y}}\mathbf{p}^{X}
\end{equation}

$\mathbf{q}^{Y}$ is a linear function of $\mathbf{p}^{X}$, thus $H(\mathbf{p}^{Y})$ is a concave function of $\mathbf{p}^{X}$. The mutual information can be expressed as:

\begin{equation}
I(X;Y) = H(\mathbf{p}^{Y}) - \sum^{n}_{i}\sum^{m}_{j} q^{Y|X}_{j} p^{X}_{i} \log q^{Y|X}_{j}
\end{equation}

The second term is a linear function of $\mathbf{p}^{X}$, hence, the whole expression is concave on $\mathbf{p}^{X}$.

\end{proof}

\section{Simplest Scenario for Communication}

In the simplest case of information transfer, at least three stages can be identified: the source of information (or transmitter), the channel over which messages are sent, and the sink (or receiver). The source is modeled as a probability space $(\Omega, \mathcal{A}_{\Omega},\mu)$. Typically, every outcome of the source will have to be processed in order to build a suitable message which can be sent over the channel. This is mathematically represented by a measurable function from the source's emitted messages to a given alphabet (usually a binary alphabet), and is practically called \emph{coding}. Conversely, in order to transmit the original information to the sink, similar functions ought to be defined on the alphabet of the received messages to the original alphabet (on the assumption that transmitter and receiver share the same language). This involves statistical estimation and is called \emph{decoding}.

The functions $f_{E}, g_{E}, f_{D}$ and $g_{D}$ are measurable functions, so $(\mathcal{W}, \mathcal{A}_{\mathcal{W}})$ and $(\mathcal{\hat{U}}, \mathcal{A}_{\mathcal{\hat{U}}})$ can be viewed as probabilizable spaces with probabilities induced by $p(W=w_{i}) = p(f_{E}^{-1}(w_{i})=u)$, and so on. This notion of inherited probability is fundamental because relative entropy is a function defined on probability simplices.

\begin{figure}[h]
\centering
  \includegraphics[scale=0.4]{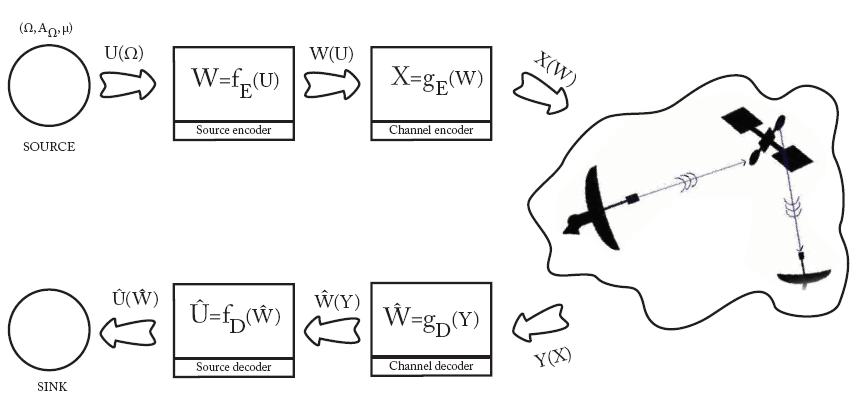}\\
  \caption{Simplest Scenario for Communication: In the sender-receiver scheme, the messages randomly emitted by the source are first compressed at the source encoder and then fed to the channel encoder. Channel encoder will map them to codewords resilient against channel noise, so that the original compressed message can be recovered reliably at the channel decoder. Source decoder will decompress the messages and deliver them to the receiver, or information sink}\label{fig1}
\end{figure}

Remarkably, the process of coding and decoding is absolutely deterministic. Choice is introduced at two levels, of different nature. First one is in the source itself, where the sample space could be \emph{whatever}, i.e. all the thoughts of a person talking on the phone. A random variable U, defined on $\Omega$ and taking values in $\mathcal{U}$ ($\|\mathcal{U}\|=n$) represents the physical resulting messages the that are emitted by the source, i.e. a series of phonemes which are a function of the thoughts of the person who talks. At a second stage, uncertainty is introduced in the channel, and is related to the noise (fading, interference, outages...) that every physical channel induces in an information carrier. In fact, a channel is represented by the tuple $(\mathcal{X},T_{Y|X},\mathcal{Y})$, where $\mathcal{X}$ and $\mathcal{Y}$ are the input and output alphabets, respectively, and $T_{Y|X}$\footnote{A generalized transition matrix would be of the form $T_{Y^{m}|X^{m-1}Y^{m-1}}$ . Here will refer only to discrete memoryless channels without feedback.} is a stochastic transition matrix such that $\mathbf{q}^{Y} = T_{Y|X}\mathbf{p}^{X}$.

Two main questions arise in this context, and that is all Information Theory is concerned about:

\begin{description}
  \item[Channel Capacity] What is the \emph{maximum} rate that can be achieved in sending information over a channel? This question is practically approached in the design of channel encoders-decoders. Later we shall see that
      $$R \leq \max_{p^{X}} I(X;Y)$$
      Error-correcting codes are mainly devoted to maximize this rate.
  \item[Rate-Distortion Theory] What is the \emph{minimum} rate at which one source can be compressed (that is, eliminate redundant parts of the source's outputs) while keeping received messages below a distortion threshold D?
      $$R \geq \min_{p^{\hat{U},U}:d(\hat{U}|U)\leq D} I(\hat{U};U) $$
      In the simplified case where the channel is noiseless, or whenever it is possible to estimate perfectly $\hat{U}$ (or just assume that  $d(\hat{U},U) = 0$), Rate-Distortion reduces to Lossless Data Compression:
      $$I(\hat{U};U) = I(U;U) = H(U)$$
      and the inequality becomes $R \geq H(\mathbf{p}^{U})$
\end{description}

\begin{figure}[h]
\centering
  \includegraphics[scale=0.4]{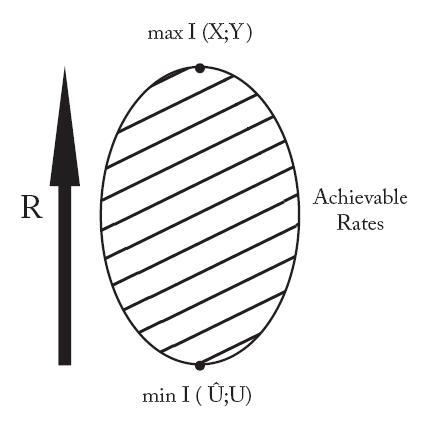}\\
  \caption{All information theory is concerned about: The rate of information transfer upper-bounded by the maximum capacity of inference of the receiver, which is related to the mutual information. On the lower part of the scheme we see that no rate is possible below those allowed for a given distortion threshold. For a discrete noiseless channel, the distortion can be taken to be zero and the lower bound reduces to the entropy of the source.}\label{fig2}
\end{figure}

There exist a nice duality between these two problems that can be appreciated when they are expressed as optimization programs \cite{chiang2004gpd}.

\section{Asymptotic Equipartition Property}

Typically sources will emit more than one output. Thus, we need to characterize them as \emph{stochastic processes} rather than as just random variables. Consider a source described by $(\Omega, \mathcal{A}_{\Omega},\mu)$ and $T:\Omega\rightarrow\Omega$ which plays the role of a time shift in the sample space. This is a dynamical system and one can derive a stochastic process from it $U_{j}(T^{j}\omega) = u_{j}, w\in\mathcal{A}_{\Omega}$.

If for all $w\in\mathcal{A}_{\Omega}$, we have that $\mu(T^{j}\omega) = \mu(\omega) = 1$ or $0$, then the source is ergodic and stationary, and Birkhoff's Theorem holds \cite{gray1990eai}:

\begin{equation}
    \lim_{m\rightarrow\infty}\frac{1}{m}\sum^{m}_{j=1} U_{j} \stackrel{\mathcal{P} = 1}{\longrightarrow} E[U_{1}] = \int U_{1} d\mu
\end{equation}

where $\mathcal{P} = 1$ denotes convergence \emph{with probability 1}. If we now consider the sequence $\{U_{j}\}^{m}_{j=1}$ and regard $\log\frac{1}{p^{U^{j}}} = -\sum^{m}_{j=1}\log p^{U_{j}|U^{j - 1}}$ as a random variable itself\footnote{$\mathbf{p}^{X}$ in boldface denotes a probability distribution, while $p^{X}$ will denote the probability of a particular occurrence of $X$,  $p(X=x)$}\footnote{$U^{j}$, with upper index, is a shorthand for the sequence $U_{1}U_{2}...U_{j}$}, function of $U^{j}$, then:

\begin{equation}
    \lim_{m\rightarrow\infty}-\frac{1}{m}\sum^{m}_{j=1} \log p^{U_{j}|U^{j - 1}}\stackrel{\mathcal{P} = 1}{\longrightarrow} E[-\log p^{U_{j}|U^{j - 1}}] = H(\mathcal{U})
\end{equation}

$H(\mathcal{U})$ is the \emph{entropy rate} of the stochastic process. It can be interpreted as the entropy of the last random variable in the sequence given all its past. Finally we get:

\begin{equation}
    \lim_{m\rightarrow\infty}-\frac{1}{m}\log p^{U^{m}} \stackrel{\mathcal{P} = 1}{\longrightarrow} H(\mathcal{U})
\end{equation}

Since convergence with probability 1 implies convergence \emph{in probability}, it is possible to write:

\begin{equation}
    \mathcal{P}(|-\frac{1}{m}\log p^{U^{m}} - H(\mathcal{U})|\leq \epsilon) \stackrel{m\rightarrow\infty}{\longrightarrow} 1, \forall\epsilon
\end{equation}

whence we obtain:

\begin{equation}
    2^{-m(H(\mathcal{U}) - \epsilon)} \leq p^{U^{m}} \leq 2^{-m(H(\mathcal{U}) + \epsilon)}
\end{equation}

Hence, for a fixed probability $\mathbf{p}^{U}$, the most likely sequences have an empirical entropy arbitrary close to the true entropy. Practically all probability mass will be localized at a proper subset of the set of all possible output sequences. This characteristic of the sequences, direct consequence of ergodicity, is called \emph{Asymptotic Equipartition Property} because as m (the length of the sequence) grows, most likely sequences tend to be grouped in a proper subset called the \emph{Typical Set} $\mathcal{T}^{m}_{\epsilon}$, whose cardinality is $2^{mH(\mathcal{U})-\epsilon} \leq ||\mathcal{T}|| \leq 2^{mH(\mathcal{U})+\epsilon}$, and gather almost all probability ($\mathcal{P}(\mathcal{T}) = 1-\epsilon$), whereas unlikely sequences tend to have a vanishing probability. Also, as m tends to infinity, all typical sequences become equally probable\footnote{This is in analogy with ensembles of statistical mechanics, where all points in phase space are assumed to be equally likely}.

For simplicity, we will restrict ourselves to stationary, independent, identically distributed (i.i.d.) processes, in which case the entropy rate takes the form:

\begin{equation}
    H(\mathcal{U}) = \lim_{m\rightarrow\infty}\frac{H(\mathbf{p}^{U^{m}})}{m} = H(\mathbf{p}^{U})
\end{equation}

which can be interpreted as the \emph{entropy per symbol} of m random variables. Finally we come to a weak version of the Asymptotic Equipartition Property. :

\begin{equation}
    \lim_{m\rightarrow\infty}-\frac{1}{m}\log p^{U^{m}} \stackrel{i. p.}{\longrightarrow} H(\mathbf{p}^{U})
\end{equation}

now convergence is in probability.

Similarly, it is possible to define \emph{jointly typical} sequences $(X^{m},Y^{m})$ with respect to a joint probability distribution $p^{X,Y}$ as the sequences for which:

\begin{equation}
    \lim_{m\rightarrow\infty}-\frac{1}{m}\log p^{X^{m}} \stackrel{i. p.}{\longrightarrow} H(\mathbf{p}^{X})
\end{equation}

\begin{equation}
    \lim_{m\rightarrow\infty}-\frac{1}{m}\log p^{Y^{m}} \stackrel{i. p.}{\longrightarrow} H(\mathbf{p}^{Y})
\end{equation}

\begin{equation}
    \lim_{m\rightarrow\infty}-\frac{1}{m}\log p^{X^{m}Y^{m}} \stackrel{i. p.}{\longrightarrow} H(\mathbf{p}^{XY})
\end{equation}

As before, in the asymptotical limit, only typical pairs will take place:

\begin{equation}
    \mathcal{P}(|-\frac{1}{m}\log p^{X^{m}Y^{m}} - H(\mathbf{p}^{XY})|\leq \epsilon)\stackrel{m\rightarrow\infty}{\longrightarrow} 1, \forall\epsilon
\end{equation}

\section{Shannon's Source and Channel Coding Theorems}

In his foundational paper \cite{shannon1948mtc}, Shannon laid the basements of Information Theory. He stated both problems above exposed (source and channel coding) and first offered a solution. For this, he used the concept of \emph{random coding}, which is not to be understood as random map between alphabets, but rather as a proof of existence of at least one coding scheme that attains the bound. However, his derivations were based on (weak) typicality and were only asymptotically optimal, therefore being of little interest until practical codes were found which performed close to the limit.

\subsection{Source Coding}

Consider a source that generates a random sequence of outputs $\{U_{j}\}^{m}_{j=1}$, and an encoding function:

$$f_{E} : U^{m} \rightarrow W(u^m)$$

with $W\in\mathcal{W} = \{1, 2, ..., 2^{mR}\}$. Here message $W$ is indexed by the instantiation of the sequence $u^m$. The cardinality of $\mathcal{W}$ will be $\|\mathcal{W}\| = 2^{mR}$, where $R$ is the rate of the code. Most commonly $\mathcal{W}\subseteq\{0,1\}^{*}$ and $W$ will be a sequence of bits. $\mathcal{W}$ is the codification of the source.

In order to quantify the \emph{fidelity} of the code one should follow one of following criteria:

\begin{itemize}
  \item $d(\hat{U},U)\leq \epsilon, \forall\epsilon$
  \item $\mathcal{P}(\hat{U}=U)\geq 1-\epsilon, \forall\epsilon$
\end{itemize}

We will use the second one, which is best suited for derivations based on weak typicality.

\begin{figure}[h]
\centering
  \includegraphics[scale = 0.4]{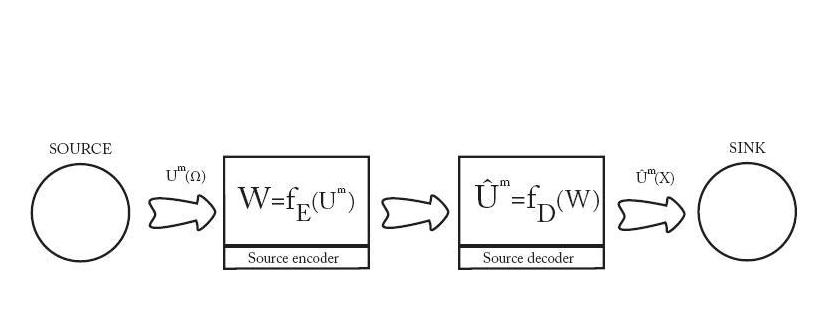}\\
  \caption{Source Coding: The sequences emitted by the source will typically have a redundant part, due to possible correlations between symbols or strings. These redundant parts don't contain much information and it is desirable to get rid of them so that no so many channel uses are required to transmit the source. This redundancy is quantified by the entropy rate of the source, since sequences of length $m$ are mapped (on average) to sequences of length $mH(U)$, which will be typically shorter.}\label{fig3}
\end{figure}

\begin{thm}\textcolor[rgb]{0.00,0.00,1.00}{[Lossless Source Coding]}An i.i.d. source $\{U_{j}\}^{m}_{j=1}$ can be reliably compressed (with vanishing probability of error) at a rate $R$ if and only if $R \geq H(\mathbf{p}^{U})$.
\end{thm}

The proof is based on asymptotical expressions, so it will be optimal in the limit $n\rightarrow\infty$

\begin{proof}[$\Rightarrow$Proof of Achievability] An error occurs whenever one of the following events happen:

\begin{itemize}
  \item The sequence is not typical: $E_{0} = \{ U^{m}\not\in \mathcal{T}^{m}_{\epsilon} \}$
  \item A codeword is indexed by more than one typical sequence\footnote{This is a consequence of random coding: in choosing a code at random we risk of selecting a bad code.}: $E_{1} = \exists u'^{m}\neq U^{m}, u'^{m}\in\mathcal{T}^{m}_{\epsilon} : W(U^{m}) = W(u'^{m})$
\end{itemize}

Using the independence bound, the error probability is:

\begin{eqnarray}
  P_{e} &\leq& P\{E_{0}\} + P\{E_{1}\}\nonumber\\
  &\leq& \epsilon +  \sum_{U^{m}\in\mathcal{T}^{m}_{\epsilon}} \sum_{\substack{u'^{m}\in\mathcal{T}^{m}_{\epsilon}\\ u'^{m}\neq U^{m}}}p^{U^{m}} P\{ W(U^{m}) = W(u'^{m})\} \nonumber\\
  &\leq& \epsilon + \sum_{U^{m}} \sum_{u'^{m}\in\mathcal{T}^{m}_{\epsilon}} p^{U^{m}} 2^{-mR}\nonumber\\
  &=& \epsilon +  \sum_{u'^{m}\in\mathcal{T}^{m}_{\epsilon}} 2^{-mR} \sum_{U^{m}}p^{U^{m}}\nonumber\\
  &=& \epsilon +  \sum_{u'^{m}\in\mathcal{T}^{m}_{\epsilon}} 2^{-mR}\nonumber\\
  &\leq& \epsilon + 2^{m(H(\mathbf{p}^{U}) + \epsilon)} 2^{-mR}\nonumber
\end{eqnarray}

Second inequality is obtained using the independence bound and averaging over all typical sequences. Also, $P\{E_{0}\}\rightarrow 0$ as $m$ grows. Third inequality is obtained by enlarging the range of the sums. Note that, for equiprobable codewords, the likelihood that two are indexed by the same sequence is $2^{mR}$. Last inequality follows from typicality arguments. Thus, in the asymptotic limit where $m\rightarrow\infty$, if:

$$R \geq H(\mathbf{p}^{U}) + \epsilon$$

the probability of error vanishes.
\end{proof}

\begin{proof}[$\Leftarrow$Weak Converse] For codes with asymptotically vanishing probability of error, the rate must necessarily satisfy $R \geq H(\mathbf{p}^{U})$. To this aim, we will make use of \emph{Fano's inequality}, which relates the probability of error to the conditional entropy of a sequence $U^{m}$ given its associated codeword $X$. It can be easily derived, so we don't prove it here:

\begin{equation}\label{fano}
    H(U^{m}|W)\leq mP_{e}\log\|\mathcal{U}\| + 1 = m\epsilon_{m}
\end{equation}

where $\epsilon_{m} = P_{e}\log\|\mathcal{U}\| + \frac{1}{n}\rightarrow 0$ as $m$ grows.

\begin{eqnarray}
  mR &\geq& H(\mathbf{p}^{W}) \nonumber\\
  &=& I(U^{m};W) + H(W|U^{m}) \nonumber\\
  &=& I(U^{m};W) \nonumber\\
  &=& H(\mathbf{p}^{U^{m}}) - H(U^{m}|W) \nonumber\\
  &\geq& mH(\mathbf{p}^{U}) - m\epsilon_{m}\nonumber
\end{eqnarray}

First inequality comes from the upper bound of entropy. Since knowing $U^{m}$ eliminates the uncertainty about $W$, we have the second equality. In the second inequality we have used Fanno's inequality. The source is modeled by an i.i.d process so $H(\mathbf{p}^{U^{m}}) = mH(\mathbf{p}^{U})$.
\end{proof}

\subsection{Channel Coding}

A channel is characterized by a the tuple $(\mathcal{X},T_{Y|X},\mathcal{Y})$, where is $T_{Y|X}$ a map between the probability simplices corresponding to the input and output alphabet. While source coding is aimed at eliminating redundant parts of source's output (for this reason named data compression), the goal of channel coding is to introduce some redundancy in a controlled way, such that it helps to fight the errors induced by the channel, and is suitably called error-correction.

Let $g_{E}$ be a channel encoding function:

$$g_{m} : W \rightarrow X^m(w)$$

here $W\in\mathcal{W} = \{1, 2, ..., 2^{mR}\}$. Each codeword $X^{m}$ is indexed by a message as before, and usually $\mathcal{X}^{m}\subseteq\{0,1\}^{*}$.

The \emph{capacity} of a discrete memoryless channel without feedback is defined:

\begin{equation}\label{capacity}
    C = \max_{\mathbf{p}^{X}} I(X;Y)
\end{equation}

and it is an upper bound on the attainable rates at which communication can take place.

\begin{figure}[h]
\centering
  \includegraphics[scale=0.4]{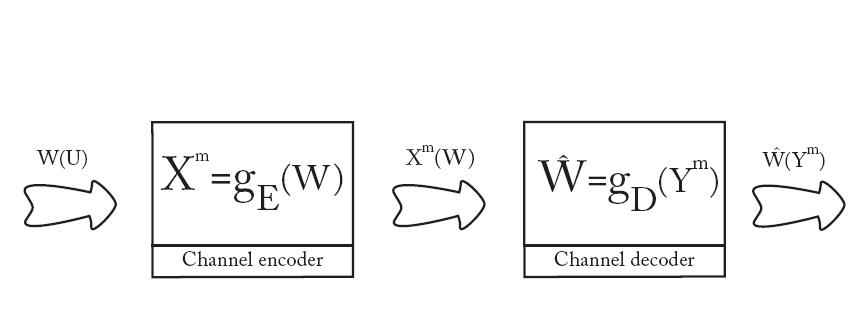}\\
  \caption{Channel Coding: At this stage, the compressed sequences are steered into larger sequences by means of introducing redundancy. The point is that, whereas the redundancy of the source's outputs was of little use, the overhead introduced by the channel encoder can be used to recover the original message even if it is corrupted by noise (but not too much). The mapping that the channel encoder performs receives the name of error-correcting code.}\label{fig4}
\end{figure}

\begin{thm}\textcolor[rgb]{0.00,0.00,1.00}{[Channel Coding]} A channel $(\mathcal{X},T_{Y|X},\mathcal{Y})$ can be used to transmit information reliably if and only if $R\leq C$
\end{thm}
\begin{proof}[$\Rightarrow$Proof of Achievability]The probability of error, averaged over all possible codes $\mathcal{C}$, is:

\begin{eqnarray}
  P_{e} &=& \sum_{\mathcal{C}}p(\mathcal{C})P_{e}(\mathcal{C}) \nonumber\\
   &=& \sum_{\mathcal{C}}p(\mathcal{C}) 2^{-mR}\sum^{2^{mR}}_{w=1}\lambda_{w}(\mathcal{C})\nonumber\\
   &=& \sum_{\mathcal{C}}p(\mathcal{C}) \lambda_{1}(\mathcal{C})\nonumber\\
\end{eqnarray}

Here $\lambda_{w} = P\{\hat{W} \neq w | W = w\}$ is the conditional probability of error given that message $w$ was sent. A random choice of code $\mathcal{C}$ symmetrizes the probabilities. Thus we will only need to consider the error probability for one codeword. Consider the event:

$$E_{w} = \{(X^{m}(w),Y^{m})\in\mathcal{T}^{m}_{\epsilon}\}$$

that is, both sequences are jointly typical. There are about $2^{mH(X,Y)}$ such pairs of sequences. The probability of error can then be expressed as:

\begin{eqnarray}
  P_{e} &=& P\{\tilde{E}_{1}\cup\bigcup^{2^{mR}}_{i=2}E_{i}\} \nonumber\\
  &\leq& \epsilon + \sum^{2^{mR}}_{i=2}P\{E_{i}\} \nonumber\\
  &=& \epsilon + \sum^{2^{mR}}_{i=2} \frac{2^{mH(X|Y)\pm \epsilon}}{2^{mH(X)\pm \epsilon}} \nonumber\\
  &\leq& \epsilon + (2^{mR} - 1)2^{m(H(X|Y) - H(X) - 2\epsilon)} \nonumber\\
  &\leq& \epsilon + 2^{m(R - I(X;Y) - 2\epsilon)} \nonumber\\
\end{eqnarray}

$\tilde{E}_{1}$ is the complementary event of $E_{1}$, and its probability vanishes as $m$ grows. The second equality is obtained from joint typicality arguments: For a given output sequence $Y^{m}$, there are about $2^{mH(X|Y)}$ jointly typical input sequences $X^{m}$. Since there are about $2^{mH(X)}$ codewords, the probability that two different codewords are jointly typical with a received sequence is  $2^{-m(I(X;Y) \pm 2\epsilon}$. Thus, the error probability will tend to zero as long as $R < I(X;Y) + 2\epsilon$.
\end{proof}

\begin{proof}[$\Leftarrow$Weak Converse] Once again, we will make use of Fano's inequality (see \ref{fano}), but now the roles are somewhat interchanged:

\begin{equation}\label{fanno2}
H(W|Y^{m})\leq mP_{e}R + 1 = m\epsilon_{m}
\end{equation}

Assuming that the messages $W$ are equiprobable:

\begin{eqnarray}
  mR &=&  H(\mathbf{p}^{W})\nonumber\\
  &=& I(W;Y^{m}) + H(W|Y^{m}) \nonumber\\
  &\leq& I(X^{m};Y^{m}) + H(W|Y^{m}) \nonumber
\end{eqnarray}

\begin{eqnarray}
  &\leq& I(X^{m};Y^{m}) + m\epsilon_{m} \nonumber\\
  &=& H(\mathbf{p}^{Y^{m}}) - H(Y^{m}|X^{m}) + m\epsilon_{m}\nonumber\\
  &\leq& \sum^{m}_{i=1}[H(\mathbf{p}^{Y_{i}}) - H(Y_{i}|X_{i})] + m\epsilon_{m}\nonumber\\
  &=& \sum^{m}_{i=1} I(X_{i};Y_{i}) + m\epsilon_{m}\nonumber\\
  &\leq& mC + m\epsilon_{m}\nonumber
\end{eqnarray}

The first inequality comes from the fact that the information contained in $Y^{m}$ about $W$ should be less or equal to the information that $Y^{m}$ contains about $X^{m}$ since $X^{m}$ is a function of $W$. Second inequality comes from Fanno's inequality. Third one comes from the independence bound. The fourth one comes from the definition of capacity (\ref{capacity}), as the maximum attainable mutual information. So as $m$ grows, the probability of error goes to zero, $m\epsilon_{m}\rightarrow 0 $, and then we have that $R\leq C$

\end{proof}

\chapter{Quantum Mechanics as a Statistical Theory}

Physical Theories deal with observable features of Nature. For a theory to be accepted, it must be capable of predicting the outcomes of experiments and phenomena within its logical framework. Otherwise they are obliged to dwell the realm of mathematical games. This implies that any theory must account for measurements, that is, besides describing Nature, it must describe how we obtain knowledge from Nature. Since scientific theories rely on evidence for justification, this should be done on a statistical basis. Measurements are subject to statistical fluctuations, although several theories obviate this fact due to the invariant nature of their observations, such as astronomy. However, in general, any theory ought to include a complete \emph{statistical model} that allows to infer system properties from measurement outcomes.

A statistical model is a part of any theory, and it consists of:

\begin{description}
  \item[Preparations] This refers to the states of the systems under consideration, like the setup of an experiment, which in classical theories are directly related to a point in phase space.
  \item[Measurements] Procedures by which physicists glean information about the systems from obtained data, which are obviously correlated to its state.
\end{description}

Mathematically this pair is denoted $(\mathfrak{S},\mathfrak{M})$, where $\mathfrak{S}$ is the set of all possible preparations and $\mathfrak{M}$ is the set of all possible measurements on these preparations.

There may be, and there usually is, uncertainty associated to both preparations and measurements. What makes Quantum Mechanics \emph{different} is indeed which these kinds of uncertainty are.

Most contents of this chapter can be found in \cite{holevo1982paq}\cite{peres1995qtc}\cite{helstrom1968sts}\cite{nielsen2000qca}.

\section{Quantum Formalism}

In Quantum Mechanics a state is defined as an \emph{equivalence class of preparations}. This means that two states are to be considered equivalent if their preparations lead to parallel vectors in state space\footnote{Here, state space is a Hilbert space $\mathcal{H}$ where vectors $|\psi\rangle\in\mathcal{H}$ represent preparations. Two vectors are equivalent if they are parallel, that is, if they are the same up to a proportionality constant. For this reason, at a basic level we will identify states with rays in Hilbert space, rather than vectors.}

Quantum Mechanics arises classically as a probabilistic theory, due to a very fundamental property of sub-microscopic systems, known as the \emph{Superposition Principle}, by virtue of which a quantum system may find itself in a complex linear combination of states. This is the hallmark of Quantum Mechanics. This property, together with the definition of state in previous paragraph, leads to a statistical model where the set of preparations is \emph{strongly convex}, in contrast to classical statistical models, where they are just convex.

The outcome of a measurement will depend probabilistically on the respective weights of the superposed states. This demands that experimenters be able of obtaining statistical ensembles of the same state in order to contrast experimental data with theoretic predictions. This automatically leads to two different (but closely interrelated) notions of probability. First one is related to the fundamental behavior of the sub-microscopic world, and second one (somewhat more classical) concerns the distribution of ensembles.

\subsection{Set of States is Convex}

The need for both quantum uncertainties and classical ensembles is best met within the \emph{C*-Algebra} formalism. A C*-Algebra $\mathfrak{C}$ is a Banach\footnote{Loosely stated, a Banach space is a Hilbert space where orthogonality is not necessarily defined.} space with unit $\mathbf{1}$ and a *-involution such that:

\begin{equation}
    \|AB\| \leq\|A\|\|B\|
\end{equation}

\begin{equation}
    \|A\|^{2}=\|A^{*}\|^{2}=\|AA^{*}\|
\end{equation}

with $A,A^{*},B \in\mathfrak{C}$. $A^{*}=A^{\dagger}$ stands for the adjoint of $A$, meaning that the algebra is closed under the adjoint operation.

Every C*-Algebra can be seen as a *-subalgebra of the algebra of bounded operators on a Hilbert space $\mathcal{H}$, $\mathcal{B}(\mathcal{H})$ \cite{bratteli1996oaa}, so it inherits the inner product:

\begin{equation}
    \langle A,B\rangle = Tr(A^{\dagger}B)
\end{equation}

Now consider the algebra $\mathfrak{A}\subseteq\mathcal{B}(\mathcal{H})$. A state $\varrho$ is a positive linear functional on this subalgebra, that maps elements in the positive cone $\mathfrak{A}_{+}$ of $\mathfrak{A}$ to nonnegative real numbers. We will only consider those functionals that fulfil $\varrho(\mathbf{1})=1$, for reasons to become clear in a while.

Let $A\in\mathfrak{A}_{+}$ be a positive operator, then one can establish the one-to-one correspondence $\varrho(\mathbf{A})=Tr(\hat{\varrho} A)$, where $\hat{\varrho}\in\mathfrak{A}_{+}$ is a positive, self-adjoint operator of trace one, called the \emph{density operator}. The requirement that the operator have trace one is related to a probability normalization. We will subsequently identify $\varrho$ with $\hat{\varrho}$.

Density operators will play an role analogous to probability distributions in classical probability. Whereas the a probability simplex $\mathfrak{P}_{n}$ has only $n$ vertices, each corresponding to a distribution where all the probability mass is accumulated at just one outcome, density operators live in a strongly convex set, meaning that there is an infinite number of extremal points, as a consequence of the Superposition Principle. This is depicted in fig. \ref{fig5}, where the simplex for two classical outcomes is compared with the set of all possible quantum preparations of a two-states system.

\begin{figure}[h]
\centering
  \includegraphics[scale=0.35]{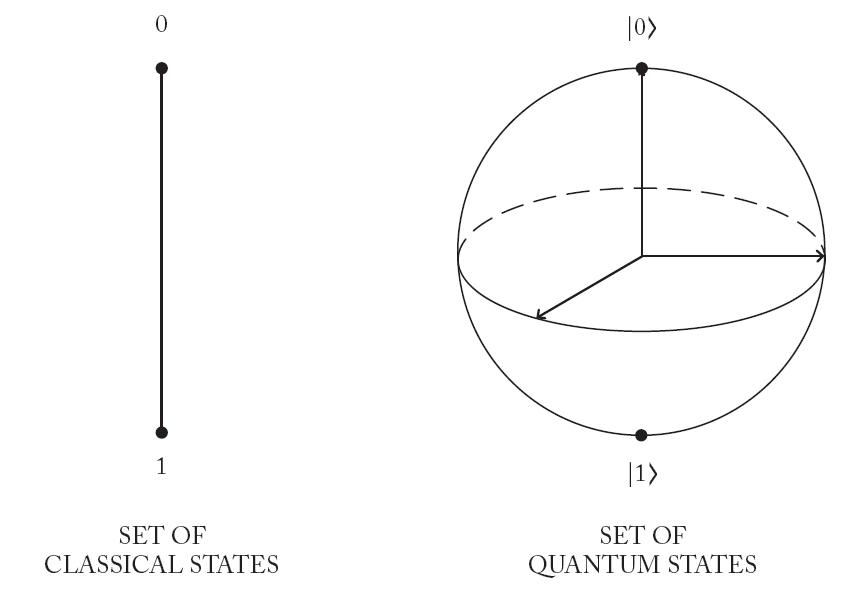}\\
  \caption{The set of quantum states is strongly convex: Classically, the set of states is given by the probability simplex $\mathfrak{P}_{2}$, any convex combination of the two attainable states, 0 and 1, remains a state. The fundamental axiom of Quantum Mechanics, the Superposition Principle, says that it is possible for a quantum bit to be in state $|0\rangle$, in state $|1\rangle$, or in a complex linear combination of both. This leads to a set of states where every superposition of states (in fact, there are infinitely many) must still be contained in the set. The set of quantum states thus is strongly convex, since there are infinitely many extremal points (living in a finite dimensional space). For a quantum bit, this set is called Bloch's ball, and the coordinates of each extremal point can be worked out from the relative phases of the pure states.}\label{fig5}
\end{figure}

The set of quantum preparations in previous figure receives the name of \emph{Bloch's ball}. Throughout this dissertation we will assume that $|0\rangle$ and $|1\rangle$ is our selected computational basis. This means that these vectors constitute a basis stable against decoherence and stand for the quantum counterpart of a bit. $|0\rangle$ and $|1\rangle$ can refer to the spin of a nucleus, the polarization of a photon, or to the state of a bistable atom. In any case, the number of degrees of freedom is two. Hence, it is the system's algebra that receives the name of \emph{qubit}, as a shorthand for quantum bit.

In general, a N-states system is described by an algebra of $N\times N$ matrices and $N^{2}$ are needed to build a basis. One suitable basis is:

\begin{equation}
    \varrho = \frac{1}{N}(\mathbf{1} + \sum^{N^{2}-1}_{i=1} r_{i}\sigma_{i})
\end{equation}

where $\{\sigma_{i}\}^{N^{2}-1}_{1}$ is some basis of self-adjoint, trace-free matrices, such that:

\begin{equation}
    \langle\sigma_{i}\sigma_{j}\rangle=\alpha\delta_{ij}
\end{equation}

\begin{equation}
    r_{i}=\frac{N}{\alpha}\langle\sigma_{i},\varrho\rangle
\end{equation}

So there is a mapping from density operators of dimension N to real vectors in $\mathbf{R}^{N^{2}-1}$. For $N=2$, this basis is the Pauli matrices and $Tr(\varrho^{2})=Tr(\varrho)=1$ if and only if $\|\mathbf{r}\|_{2} = 1$, i.e., rank one density operators lie in the boundary of Bloch's ball. A density operator having rank one is called a \emph{pure state}, and otherwise is called a \emph{mixed state}. Any mixed state can be expressed as a convex combination of pure states:

\begin{equation}
    \varrho_{mixed} = \sum_{j} q_{j}\varrho_{j}
\end{equation}

where $\mathbf{1}^{T}\mathbf{q} = 1$, and $\varrho_{j}=|j\rangle\langle j|$ are rank one density operators.

\subsection{Set of Measurements is Also Convex}

A measurement $\mathbf{M}\in\mathfrak{M}$ is an affine map from $\mathfrak{S}$ to the set $\mathfrak{P}_{\mathcal{U}}$ of all probability distributions in some probability space $(\mathcal{U},\mathcal{A}_{\mathcal{U}} , \mathbf{p}^{U})$:

$$\mathbf{M}: \mathfrak{S}\longrightarrow \mathfrak{P}_{\mathcal{U}}$$

Classically, this can reflect the statistical bias of a measuring apparatus or procedure, and amounts to a reshaping of the probability simplex. In the quantum world, a measurement is defined on a strongly convex set, where ``quantum probabilities" live, and takes values in a classical probability simplex, so forcedly some structure must be lost in a measurement process. This is sometimes called the \emph{wave-packet collapse}, or \emph{decoherence} (see section \ref{sec:decoherence}).

Consider the measurement $\mathbf{M}(\varrho)=\mathbf{p}^{U}$. We shall write $\mathbf{M} = \{M(u_{1}), M(u_{2}),\\ ..., M(u_{k})\}$ where each $M(u_{j})\in\mathfrak{A}_{+}$ is a positive operator, associated to an $u_{j}$ in $\mathcal{A}_{\mathcal{U}}$\footnote{These events need not be elementary events: as members of the $\sigma$-algebra $\mathcal{A}_{\mathcal{U}}$, they may in general be subsets of the sample space $\mathcal{U}$}. The probability of event $u_{j}$ is given by:

\begin{equation}
    p(u_{j}) = \langle\varrho, M(u_{j})\rangle
\end{equation}

A measurement is also called a \emph{Positive Operator Valued Measure} (POVM), since it relates a probability measure with an operator in the positive cone of the algebra $\mathfrak{A}$. A POVM has the following properties:

\begin{eqnarray}
  M(\emptyset) &=& 0 \\
  M(\mathcal{U}) &=& \mathbf{1} \\
  u_{i}&\subseteq& u_{j} \Longrightarrow M(u_{i})\leq M(u_{j})\\
  u_{i}&=&\bigcup_{j} u_{j} \Longrightarrow M(u_{i})=\sum_{j}M(u_{j})
\end{eqnarray}

Since it is required that the whole sample space be covered, i.e., $\bigcup^{k}_{j} u_{j} = \mathcal{U}$, then:

\begin{equation}
    \sum^{k}_{j=1}M(u_{j})=\mathbf{1}
\end{equation}

which ensures that probability is normalized $p(\mathcal{U}) = Tr(\varrho\mathbf{1}) = 1$. The $M(u_{j})$ constitute a \emph{resolution of the identity}.

Note that, even if $u_{i}\bigcap u_{j} = \emptyset$, in general it still may be the case that $M(u_{i})M(u_{j})\neq \delta_{ij}M(u_{j})$. In this case we have a non-orthogonal resolution of the identity, also known as a \emph{fuzzy measurement}.

Whenever $u_{i}\bigcap u_{j} = \emptyset \Rightarrow M(u_{i})M(u_{j})\ = \delta_{ij}M(u_{j})$ holds, we have a \emph{projective measurement}. This is justified because $M(u_{i})^{2}=M(u_{i})$ are projectors. Projective measurements are extremal points of $\mathfrak{M}$, and are also called \emph{von Neumann measurements}. The converse is in general not true: there can be non-orthogonal resolutions of the identity  at the extremal points of $\mathfrak{M}$. However, for qubits, where $\mathcal{U} =\{0,1\}$, the converse is true \cite{holevo1982paq} and one can say that a measurement is an extremal point of $\mathfrak{M}$ if and only if it is a projective measurement.

Observables are directly related to projective measurements through the \emph{Spectral Theorem}, which says that any self-adjoint operator $X$ admits the spectral representation:

\begin{equation}
    X = \sum_{u_{i}\in\mathcal{U}} u_{i}M(u_{i})
\end{equation}

where $M(u_{i})$ is an orthogonal resolution of the identity. $u_{i}$ constitute the spectrum of the observable and $M(u_{i})$ determine the eigenspace associated to each eigenvalue.

In a practical scope, it is not known how to implement a general non-orthogonal POVM, defined in a state space $\mathcal{H}_{1}$. However, Neumark's Theorem \cite{holevo1982paq} ensures that it is possible to simulate a POVM with a projective measurement defined in an extended space $\mathcal{H}_{1}\otimes \mathcal{H}_{A}$. The letter ``A" stands for ancilliary system.

\section{Von Neumann's Entropy }

Just as classical entropy is defined on a probability simplex, it is possible to define an entropy for quantum probability distributions, called the \emph{von Neumann's entropy}, which is defined on the set of quantum states:

\begin{equation}
    S(\varrho) = -Tr(\varrho\log\varrho)
\end{equation}

In the case of orthogonal states, it reduces to Shannon's entropy. In fact, many properties (but not all) of classical entropy still hold in the quantum case. We can derive them,as in previous chapter, from the more fundamental \emph{quantum relative entropy}:

\begin{equation}
    S(\varrho||\sigma) = Tr(\varrho\log\varrho - \varrho\log\sigma)
\end{equation}

\begin{thm}\textcolor[rgb]{0.00,0.00,1.00}{[Nonnegativity of quantum relative entropy]}The quantum relative entropy is positive semidefinite, $S(\varrho||\sigma) \geq 0$
\end{thm}
\begin{proof} It is possible to find a diagonalization for each density operator, $\varrho = \sum_{i} p_{i}|i\rangle\langle i|$ and $\sigma = \sum_{j} q_{j} |j\rangle\langle j|$, then:

\begin{eqnarray}
  S(\varrho||\sigma) &=&  \sum_{i} p_{i}[\log p_{i} - \sum_{j}\log q_{j}\|\langle i|j\rangle\|^{2}]\nonumber\\
   &=&  \sum_{i} p_{i}[\log p_{i} - \sum_{j}D_{i,j}\log q_{j}]\nonumber\\
   &\geq&  \sum_{i} p_{i}[\log p_{i} - \log \sum_{j}D_{i,j}q_{j}]\nonumber\\
   &=&  \sum_{i} p_{i}\log\frac{p_{i}}{ \sum_{j}D_{i,j}q_{j}}\nonumber\\
   &=&  \sum_{i} p_{i}\log\frac{p_{i}}{r_{i}}\nonumber\\
   &\geq&  0 \nonumber
\end{eqnarray}

The diagonalizations need not be equal, thus the possible overlap between the states must be accounted for. This overlap is encoded in a doubly stochastic matrix $D_{ij} = \|\langle i|j\rangle\|^{2} \geq 0$. The first inequality comes from a slight variation of Jensen's inequality and the concavity of the logarithm. The last inequality is a just a property of the classical relative entropy (see Theorem \ref{thm:relent}).
\end{proof}

\begin{thm}\label{thm:subadd}\textcolor[rgb]{0.00,0.00,1.00}{[Subadditivity of von Neumann's Entropy]} For a global system $\varrho^{AB}$ the joint entropy satisfies $S(\varrho^{AB})\leq S(\varrho^{A})+S(\varrho^{B})$ with equality if and only if both systems are uncorrelated.
\end{thm}
\begin{proof} As a consequence of the nonnegativity of the quantum relative entropy, we can write:

$$D(\varrho||\sigma) = -S(\varrho) - Tr(\varrho\log\sigma)\geq 0$$

Taking $\varrho = \varrho^{AB}$ and $\sigma=\varrho^{A}\otimes\varrho^{B}$, we obtain:

\begin{eqnarray}
  S(\varrho^{AB}) &\leq& - \varrho^{AB}\log(\varrho^{A}\otimes\varrho^{B})\nonumber\\
  &=& - \varrho^{AB}(\log\varrho^{A} + \log\varrho^{B}) \nonumber\\
  &=& - \varrho^{A}\log\varrho^{A} - \varrho^{B}\log\varrho^{B}\nonumber\\
  &=& S(\varrho^{A}) + S(\varrho^{B})
\end{eqnarray}

To see that the bound is tight if and only if the $\varrho^{AB}=\varrho^{A}\otimes\varrho^{B}$, one need only consider the relative entropy $S(\varrho^{AB}||\varrho^{A}\otimes\varrho^{B})$.
\end{proof}

\begin{thm}\textcolor[rgb]{0.00,0.00,1.00}{[Concavity of von Neumann's entropy]} Von Neumann's entropy is a concave function of $\varrho$
\end{thm}
\begin{proof}To prove this result, we will make use of a spurious system B. Consider the joint state:

$$\varrho^{AB} = \sum_{i}p_{i}\varrho^{A}_{i}\otimes|i\rangle\langle i|^{B}$$

Its von Neumann's entropy is:

\begin{eqnarray}
  S(\varrho^{AB}) &=&  S(\sum_{i}p_{i}\varrho^{A}_{i}\otimes|i\rangle\langle i|^{B})\nonumber\\
  &=& S(\sum_{i}p_{i}(\sum_{j}\lambda^{j}_{i}|j\rangle\langle j|^{A})\otimes|i\rangle\langle i|^{B}) \nonumber\\
  &=& -\sum_{i,j}p_{i}\lambda^{j}_{i}\log\lambda^{j}_{i} - \sum_{i,j}p_{i}\lambda^{j}_{i}\log p_{i}\nonumber\\
  &=& -\sum_{i,j}p_{i}\lambda^{j}_{i}\log\lambda^{j}_{i} - \sum_{i}p_{i}\log p_{i}\nonumber\\
  &=& \sum_{i}p_{i}S(\varrho^{A}_{i}) + H(\mathbf{p})\nonumber
\end{eqnarray}

where $\sum_{j}\lambda^{j}_{i}|j\rangle\langle j|^{A}$ is a diagonalization of $\varrho^{A}_{i}$. Note that $\sum_{j}\lambda^{j}_{i} = 1$. On the other hand, the entropies of the separated systems are:

$$S(\varrho^{A}) = S(\sum_{i}p_{i}\varrho^{A}_{i})$$

$$S(\sum_{i}p_{i}|i\rangle\langle i|^{B}) = H(\mathbf{p})$$

Making use of the subadditivity property proved before, we have that:

\begin{eqnarray}
  \sum_{i}p_{i}S(\varrho^{A}_{i}) + H(\mathbf{p}) &\leq& S(\sum_{i}p_{i}\varrho^{A}_{i}) + H(\mathbf{p})\nonumber\\
  \sum_{i}p_{i}S(\varrho^{A}_{i}) &\leq& S(\sum_{i}p_{i}\varrho^{A}_{i}) \nonumber
\end{eqnarray}

thus, the von Neumann's entropy is convex.
\end{proof}

\section{Classical Information and Quantum Information are Not the Same}

The difference between bits and qubits is more fundamental than just terminology. Whereas classical bits are symbolic representations of the information stored in a physical system (i.e. modulated waves, or the orientation of the magnetic cells in a hard drive...), qubits are to be \emph{identified} with physical systems, or with their algebra at least. Quantum information is more general than classical information, since the symbolic representation of information arises in the special case where only orthogonal states are considered.

\subsection{Classical Information through Decoherence}\label{sec:decoherence}

Since Quantum Mechanics supersedes classical theories, it is expected that classical probability can as well be represented in the language of density operators. Consider a probability distribution $\mathbf{p}^{X}$ over $\mathcal{X}$ ($\|\mathcal{X}\| = n$), then its operator counterpart can be written as:

\begin{equation}\label{decoher}
    \varrho^{X} = \sum^{n}_{i = 1} p_{i}|i\rangle\langle i|
\end{equation}

This density operator belongs to the algebra of diagonal matrices. In the qubit case, this algebra is the set of density operators in the segment that passes through the poles in Bloch's ball (see fig. \ref{fig5}).

Quantum states can be described by density operators having off-diagonal terms, which are responsible for quantum interferences, and this is directly related to the fact that the set of states is strongly convex. How classical properties arise from quantum-mechanical laws is a itself a topic of intense research and receives the name of \emph{Environmental Decoherence} \cite{zeh1996ad}\cite{joos1999eed}. In the information-theoretic context of this thesis, it suffices to say that in the measuring process that both the measured system and the measuring apparatus evolve together in time (according to some interaction Hamiltonian) into a \emph{preferred diagonal basis}, induced by the interaction of the measuring apparatus and their environment \cite{zurek1981pbq}.

Let $\varrho^{S}=\sum_{j}q_{j}\varrho_{j}$, $\varrho^{A}_{0}$ and $\varrho^{E}_{0}$ be the initial states of a system, an apparatus and their environment, respectively. In a first step the system and the apparatus become correlated, so that observing the apparatus will give us information about the system. In a second step, the apparatus is let alone to evolve along with its environment. This process is depicted in fig. \ref{fig6}:

\begin{equation}\label{usa}
    \mathbb{U}_{SA}(\sum_{j}q_{j}\varrho_{j}\otimes\varrho^{A}_{0})\mathbb{U}^{\dagger}_{SA} = \sum_{j}q'_{j}\varrho_{j}\otimes\varrho^{A}_{j}\rightarrow
\end{equation}

\begin{equation}\label{uae}
    \rightarrow \mathbb{U}_{AE}(\sum_{j}q'_{j}\varrho_{j}\otimes\varrho^{A}_{j}\otimes\varrho^{E}_{0})\mathbb{U}^{\dagger}_{AE} =
\sum_{j}q''_{j}\varrho_{j}\otimes\varrho^{A}_{j}\otimes\varrho^{E}_{j}
\end{equation}

The measuring apparatus and the environment become rapidly correlated, and the off-diagonal terms in the system's density operator are swept away. Provided that the environment remains in a pure state and that $\langle\varrho^{E}_{i}\varrho^{E}_{j}\rangle\approx\delta_{ij}$\footnote{This assumption basically comes from the fact that we don't observe quantum interferences between macroscopic states}, tracing out the environment and the apparatus (see section \ref{sec:marginal}) leaves us with a classical probability simplex:

\begin{equation}\label{decoher2}
    \varrho^{SA}\approx\sum_{j}q''_{j}|j\rangle\langle j|\otimes\varrho^{'A}_{j}
\end{equation}

\begin{figure}[h]
\centering
  \includegraphics[scale=0.35]{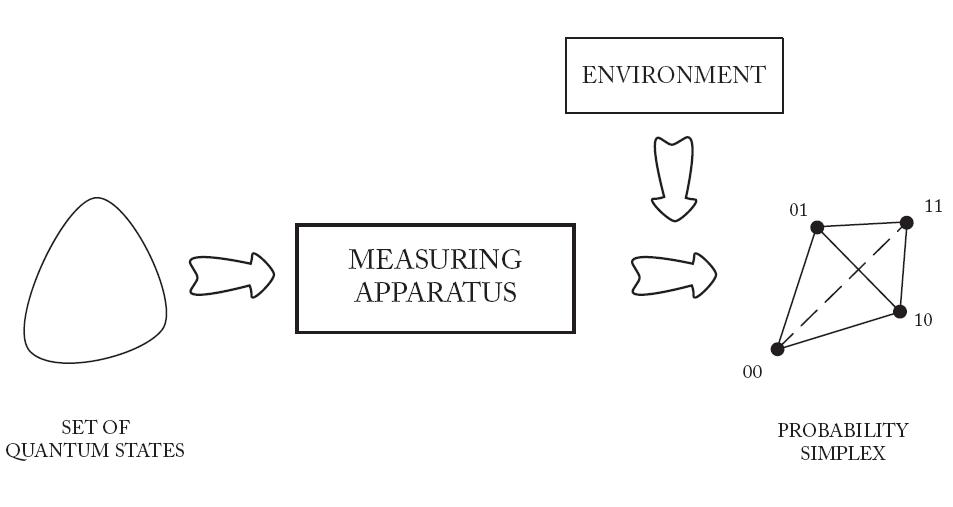}\\
  \caption{Decoherence in the measuring process: The interaction of the measuring apparatus with its environment causes the quantum correlations to dilute in the joint Hilbert space of the system, apparatus and environment. The crux of this process is that the states of the environment are by definition unaccessible, so no measurement could ever detect these correlations. Thus, locally, the joint system-apparatus state appears to be in a diagonal matrix state, as a consequence of tracing out the environment. In this picture, a measurement with four outcomes is represented.}\label{fig6}
\end{figure}

such that $[\varrho^{'A}_{j},\mathbb{U}_{AE}] = 0$, i.e. they share a diagonal basis\footnote{This means that the measuring apparatus evolves to a state which is stable against decoherence, i.e., stationary under macroscopic time evolution}. Eq. \ref{decoher2} is to be compared with eq. \ref{decoher} The quantum probabilities don't disappear, but get dispersed in the correlations between the system and its environment.

Note that the expressions \ref{usa} and \ref{uae} are not correct in general, since time evolution couples the system, the apparatus and their environment in such a way that their global state cannot be expressed as a product state: the final state may in general be \emph{entangled}. However, to illustrate the correlations that take place during the measuring process we use these partially ``allegoric" expressions. In later chapters the correctness will be restored.

\subsection{No-Cloning Theorem}\label{sec:nocloning}

Another way to see the difference between classical and quantum information is to imagine a machine capable of copying quantum states. The machine is fed at its input with an unknown quantum state and it outputs two copies of the initial state. This machine cannot exist:

\begin{thm}\textcolor[rgb]{0.00,0.00,1.00}{[No-Cloning Theorem]} It is impossible to copy unknown quantum states.
\end{thm}
\begin{proof}
Without loss of generality, we shall only consider pure states. Consider two \emph{unknown} states $\varrho_{1}$ and $\varrho_{2}$, which are fed as input into the copying machine in an initial pure state $\varrho^{CM}$. The copying process is described as a time evolution $\mathbb{U}$ of the whole system:

\begin{equation}
    \mathbb{U}(\varrho_{1}\otimes\varrho^{CM})\mathbb{U}^{\dagger} = \varrho_{1}\otimes\varrho_{1}
\end{equation}

\begin{equation}
    \mathbb{U}(\varrho_{2}\otimes\varrho^{CM})\mathbb{U}^{\dagger} = \varrho_{2}\otimes\varrho_{2}
\end{equation}

if we now take the inner product of the two equations, we have that:

\begin{equation}
    \langle\varrho_{1},\varrho_{2}\rangle=\langle\varrho_{1},\varrho_{2}\rangle^{2}
\end{equation}

thus both states must be either orthogonal, or the same. This requirement is in contradiction with the assumption of the two qubits being unknown.
\end{proof}

As expected, this is in accordance with the existence of  classical fan-out gates, which have the capacity of copying a bit as many times as desired.

\subsection{The Holevo's Bound}\label{sec:holevo}

Classically, the capacity of inference is related to the mutual information. An observer (receiver) can reliably guess the value of an experiment (or channel use) provided that $I(X;Y)$ is arbitrary close to $H(X)$. In principle, thanks to the use of better preparation and measuring devices (equivalently, coding and decoding schemes), mutual information can be brought very close to its upper bound.

Quantum mechanics prevents this fact, once again as a consequence of the Superposition Principle, because there may be states which are not orthogonal, and no measurement can, even in principle, distinguish them with 100\% reliability.

\begin{thm}\textcolor[rgb]{0.00,0.00,1.00}{[The Holevo's Bound]} Let $X\in\mathcal{X}$ be encoded in state $\varrho^{X}=\sum^{n}_{i=1}p^{X}_{i}\varrho_{i}$, where the $\varrho_{i}$ have orthogonal support, and a measurement $\mathbf{M}_{Y}(\varrho^{X})= \mathbf{p}^{Y}$, the accessible information is upper bounded by:

\begin{equation}\label{holevobound}
    I(X:Y) \leq S(\varrho^{X}) - \sum^{n}_{i=1}p^{X}_{i}S(\varrho_{i})
\end{equation}
\end{thm}

\begin{proof} Mutual information can be written as:

$$I(X:Y) = H(\mathbf{p}^{X}) - H(\mathbf{p}^{X}|Y)$$

Last term represents the uncertainty about $X$ provided that measurement $\mathbf{M}_{Y}$ was chosen:

$$H(\mathbf{p}^{X}|Y)) = \sum^{m}_{j=1}p(y_{j})\sum^{n}_{i=1}p(x_{i}|y_{j})\log p(x_{i}|y_{j})$$

with $p(x_{i}|y_{j}) = \langle \varrho_{i}, M_{j}\rangle$. It is easy to see that the conditional entropy will vanish if and only if $\langle \varrho_{i},M_{j}\rangle =\delta_{ij}$. Now, suppose that this is indeed the case: selected measurement scheme is optimal. Reasoning in a similar way to Theorem \ref{thm:subadd}, it is possible to write:

$$S(\sum^{n}_{i=1}p^{X}_{i}\varrho_{i}) = \sum^{n}_{i=1}p^{X}_{i}S(\varrho_{i}) + H(\mathbf{p}^{X})$$

The optimal measurement strategy yields:

$$I(X:Y) = H(\mathbf{p}^{X}) = S(\sum^{n}_{i=1}p^{X}_{i}\varrho_{i}) - \sum^{n}_{i=1}p^{X}_{i}S(\varrho_{i})$$

For measurements that are not optimal, we will have in general that:

$$I(X:Y)\leq S(\sum^{n}_{i=1}p^{X}_{i}\varrho_{i}) - \sum^{n}_{i=1}p^{X}_{i}S(\varrho_{i})$$

\end{proof}

Note if the states $\varrho_{i}$ are chosen to be pure, the upper bound in \ref{holevobound} reduces to the classical entropy. One direct conclusion to be drawn from previous Theorem is that the information contained in a qubit is, at most, one bit. This discouraging result may lead us to the opinion that quantum information has no real advantages over classical information. As we will see in next part, this belief is wrong.

\section{Experiments as Information Transfer}

Perhaps it is illuminating to see that it is possible, just with a slight change in the terminology, to compare the two main scenarios that are occupying us: a communications channel, and a physical experiment.

\begin{figure}[h]
\centering
  \includegraphics[scale=0.4]{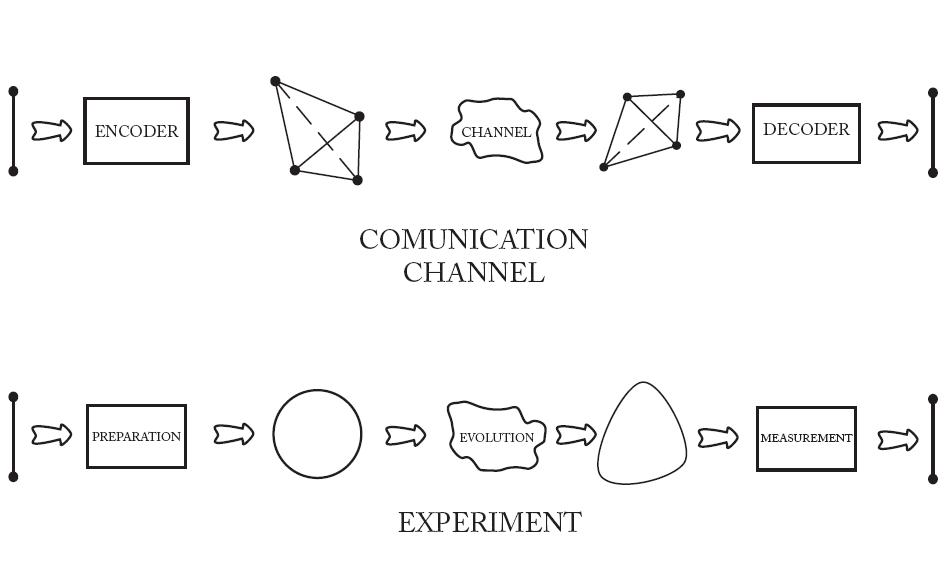}\\
  \caption{Pictorial duality between experiments and channels: In the communication channel scenario, we represent a source taking values in two symbols, encoded in a four symbol code. Then, stochastic evolution takes place and the probability simplex is distorted (dimension can increase, decrease or remain the same, and symbol frequencies may change). At the decoder, the original two symbols should be recovered with their original frequencies. In the case of experiments, a quantum state is prepared. According to a known Hamiltonian, the system undergoes a deterministic time evolution jointly with the unknown system of which information is to be obtained. Then several hypothesis are encoded in the set of states, as a result of the joint time evolution.}\label{fig7}
\end{figure}

The main goal of both operations is to gain information about the state of an unaccessible system. In a communications channel this system is the source. In the experiment, it is an unknown system that is forced to interact with another one, previously prepared in a quantum state. Thus in the experiment scenario, information enters at the evolution stage. Encoding is a procedure analogous to a preparation. Whereas temporal evolution is totally deterministic, the channel is of stochastic nature, but it still represents some sort of time flow. After undergoing a temporal evolution or a channel, the sets of states are distorted to some extent. Finally, both measurement and decoding entail the estimation of a probability distribution out of the incoming sets of states. In the case of experiments, uncertainty is introduced at this stage, if the states are not orthogonal.

\part{Quantum Entanglement: A New Resource for Communication}
\chapter{Quantum Non-locality}

To jointly describe multiple systems, an algebra is needed which contains the elements of global measurements, as well as those part of partial (marginal) measurements. For bipartite systems, the algebra of the composite system is:

$$\mathfrak{C} = \mathfrak{A}\otimes\mathfrak{B}$$

where $\mathfrak{A}$ and $\mathfrak{B}$ are the operator algebras of the subsystems, respectively. For the algebra of diagonal matrices of dimension $N$, that is, for classical distributions of $N$ different outcomes, the number of orthogonal matrices needed to form a basis is $N$. For two classical systems of the same dimension, we will need $N^{2}$ such matrices. In the quantum case, since matrices need not be diagonal, the number of matrices that form a basis is $N^{2}$, and for two systems $N^{4}$ matrices would be needed.

For classical systems, if two different observers perform a measurement, each one at a different system, they will each gain $\log N$ bits of information. If they combine their information about the subsystems, it will be possible to reconstruct the global state, as it only demands $2\log N$ bits.

In the quantum case, the Holevo's bound says that each observer can gain \emph{at most} $\log N$ bits. Thus, there is no way no learn about the global state just from the marginal measurements, for it demands $4\log N$ bits, whereas there are only $2\log N$ available.

This suggests that there is more information contained in the composite system, than in the sum of the informations contained in its components. This characteristic of quantum systems gives rise to a new fundamental phenomenon called \emph{quantum non-locality}.

\section{EPR Paradox}

In their famous paper \cite{einstein1935cqm}, Einstein, Podolsky and Rosen (EPR) came to the conclusion that Quantum Mechanics was in awkward epistemological status, due to the its lack of at least one of the properties required to any theoretical framework which intends to describe Reality. This properties can be stated as two principles:

\begin{description}
  \item[Principle of Locality] Two causally disconnected, i.e. spatially separated,  measurements cannot exert any influence on one another.
  \item[Principle of Realism] Any physical theory must account for every element of reality, this meaning that every possible outcome of an experiment should have a definite value prior to its measurement.
\end{description}

EPR showed that Quantum Mechanics violates at least one of these two principles, so a quantum description of Reality cannot be completely accurate. In a gedanktexperiment devised by Bohm \cite{bohm1952siq}, which involves particles of spin one half.

Consider the pure state of a composite system $\varrho^{AB}_{\psi}\in\mathfrak{A}\otimes\mathfrak{B}$, $\varrho^{AB}_{\psi} =|\psi\rangle\langle\psi|$, where:

\begin{equation}
    |\psi\rangle = \frac{1}{\sqrt{2}}(|0\rangle^{A}\otimes|1\rangle^{B} - |1\rangle^{A}\otimes|0\rangle^{B})
\end{equation}

is a state vector representing the preparation of two two-states systems. This state may be created as pairs of photons of opposite polarization emitted from a common source\footnote{Photons are massless spin one particles, so their polarization has just two degrees of freedom, and can be modeled as a two-state system.}. The indices A and B denote two different locations or observers, causally disconnected, where each operator algebra is defined, respectively.

Given that the source is capable of providing an unlimited number of pairs, observer at location A performs a (projective) measurements on its particles $\mathbf{M}^{A}(\varrho^{A})=\mathbf{p}^{A}$, and so does the observer at location B, obtaining the distribution $\mathbf{M}^{B}_{i}(\varrho^{B})=\mathbf{p}^{B}, i=1,2$. Observer at location B is able to choose between the different measurements:

\begin{eqnarray}
  \mathbf{M}^{B}_{1} &=& \{\left(\begin{array}{cc} 1 & 0 \\  0 & 0 \\\end{array}\right),\left(\begin{array}{cc} 0 & 0 \\ 0 & 1 \\\end{array} \right)\} \\
  \mathbf{M}^{B}_{2} &=& \{\frac{1}{2}\left(\begin{array}{cc} 1 & 1 \\  1 & 1 \\\end{array}\right),\frac{1}{2}\left(\begin{array}{cc} 1 & -1 \\ -1& 1 \\\end{array} \right)\}
\end{eqnarray}

and $\mathbf{M}^{A}=\mathbf{M}^{B}_{1}$. The two measurements correspond to different orientations of the polarized detector. From basic Quantum Mechanics it is not hard to see that these orientations are orthogonal in real space.

If A and B choose to use the same measurement setup (detectors polarized in the same direction), due to the structure of the state $\varrho^{AB}_{\psi}$, whenever A measures its particle pointing upwards, B will necessarily find it pointing downwards, and viceversa. If B uses a detector polarized in an orthogonal direction, then its outcomes will be uncorrelated to those of A, which comes from the fact that:

\begin{equation}
    \langle M^{A}_{1}, M^{B}_{2,1}\rangle = - \langle M^{A}_{1}, M^{B}_{2,2}\rangle = \langle M^{A}_{2}, M^{B}_{2,1}\rangle = - \langle M^{A}_{2}, M^{B}_{2,2}\rangle = \frac{1}{2}
\end{equation}

that is, there will always be some probability overlap between the outcomes in the different orientations.

\begin{figure}[h]
\centering
  \includegraphics[scale=0.4]{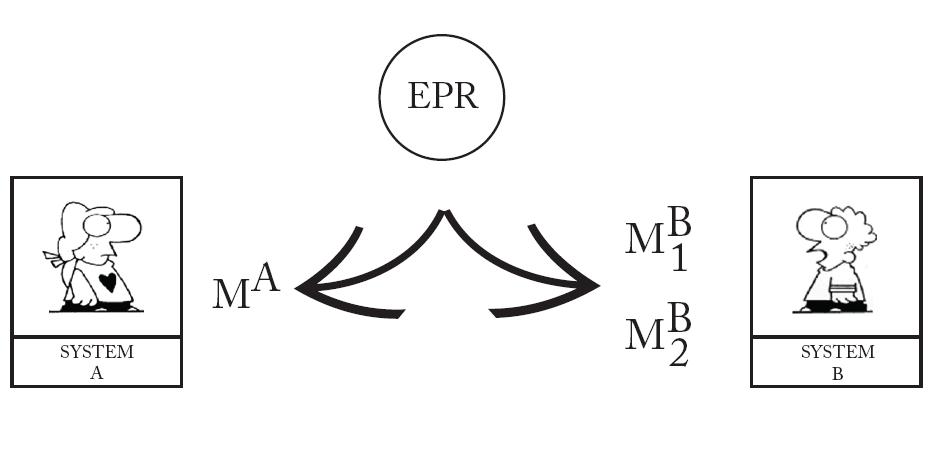}\\
  \caption{EPR paradox: Each subsystem of the EPR pair is delivered to a different observer. A and B are spatially separated, so the choice of observer B should not influence the outcomes of A.}\label{fig8}
\end{figure}

The EPR paradox can be stated as follows. Suppose that at a first stage, A and B are measuring their respective particles in different directions, i.e. using different measurement setups. Their statistics will be plain, $\mathbf{p}^{A}=\mathbf{p}^{B}=\left(\begin{array}{c}\frac{1}{2} \\\frac{1}{2} \\\end{array}\right)$ (see \ref{sec:marginal}). Now, suppose that, right before measuring its particles, B always switches to its alternative setup without letting A know about this change (they might be many lightyears apart), and measures in the same direction as A does. No matter how far apart they happen to be, if A gets $|0\rangle$, then B will get $|1\rangle$ with certainty. Their outcomes will be correlated, yet they will not be aware of this correlation unless they communicate their results, for their statistics will remain plain.

How does the particle at A learn about the change in the orientation of detector in B, despite being causally disconnected from it, is the EPR paradox. One must draw the conclusion that either:

\begin{itemize}
  \item Quantum Mechanics violates the Principle of Locality, or
  \item Quantum Mechanics is incomplete and some \emph{hidden-variable theory} that supersedes Quantum Mechanics is needed to explain these non-classical correlations.
\end{itemize}

\subsection{Marginal Measurements}\label{sec:marginal}

So far we mentioned $\varrho^{AB}$, $\varrho^{A}$ and $\varrho^{B}$ as the density operators of the whole system and of its components, respectively. The procedure to obtain the marginal density operators from the joint one is analogous as in classical probability. Let $\mathbf{p}^{AB} = \{p^{AB}_{ij}\}^{n,m}_{i=1,j=1}$ be the joint probability for variables A and B\footnote{Here we don't assume that A and B are systems with the same dimension, so n and m need not be equal}. The marginal in A is obtained via:

\begin{equation}\label{sumout}
    p^{A}_{i} = \sum^{m}_{j=1} p^{AB}_{ij}
\end{equation}

To ``sum out" one probability distribution is equivalent to ignore what is happening in the system associated to B. In Quantum Mechanics, this deliberate ignorance amounts to perform the trivial measurement, $\mathbf{M}_{B} = \{\mathbf{1}\}$, in the system we want to ignore:

\begin{equation}
    p^{A}_{i} = \langle \varrho^{AB}, M^{A}_{i}\otimes\mathbf{1}\rangle
\end{equation}

If with $\varrho^{AB}_{ab, a'b'} $ we denote the entries of $\varrho^{AB}$ corresponding to its subsystems, where $aa'$ represents the degrees of freedom localized at A and $bb'$ those localized at B, then eq. \ref{sumout} can be developed:

\begin{eqnarray}
  p^{A}_{i} &=& \sum^{n,m}_{aa',bb'}\varrho^{AB}_{ab, a'b'}M^{A}_{i,aa'}\delta_{bb'} = \nonumber\\
  &=& \sum^{n}_{aa'}(\sum^{m}_{b}\varrho^{AB}_{ab, a'b})M^{A}_{i,aa'} = \nonumber \\
  &=& \sum^{n}_{aa'} \varrho^{A}_{aa'}M^{A}_{i,aa'} = \langle \varrho^{A},M^{A}_{i} \rangle
\end{eqnarray}

where $\varrho^{A} = \sum^{m}_{b}\varrho^{AB}_{ab, a'b}$ is the so-called \emph{reduced density operator}, obtained by disregarding system B. The operation of tracing out one of the subsystems is called \emph{partial trace} of a state, and is denoted $\varrho^{A} = Tr_{B}\varrho^{AB}$.

Thus, we have seen that the subalgebras of marginal measurements can be obtained just by means of tensor-multiplying positive operators with the identity.

\section{Quantum Correlations and Bell's Inequalities}

EPR agreed on that the predictions of Quantum Mechanics were indeed correct, but ultimately explainable in terms of statistical distributions of some ``hidden variables", which would be in harmony with the principles of locality and realism. This conjecture could neither be proved nor refuted until the advent of \emph{Bell's inequalities}\cite{bell1964epr}. These inequalities have, a priori, nothing to do with Quantum Mechanics, but rather put a constraint on the correlations predictable by any theory that incorporates ``local realism" (we use this name to refer to both principles introduced above).

A bipartite system is said to be correlated if:

\begin{equation}\label{correlation}
    \mathbf{p}^{AB}\neq \mathbf{p}^{A}\mathbf{p}^{B} \Leftrightarrow \langle\varrho^{AB}, \mathbf{M}^{A}\otimes \mathbf{M}^{B}\rangle\neq\langle\varrho^{A}, \mathbf{M}^{A}\rangle\langle\varrho^{B}, \mathbf{M}^{B} \rangle
\end{equation}

for some measurements $\mathbf{M}^{A}$ and $\mathbf{M}^{B}$. Here the inner product has to be understood as componentwise products. Equation \ref{correlation} is equivalent to demand that the density operator factorize:

\begin{equation}\label{dependence}
    \varrho^{AB}\neq\varrho^{A}\otimes\varrho^{B}
\end{equation}

A quantum state may exhibit two kinds of correlations: classical and quantum. Classical correlations arise whenever a state is of the form:

\begin{equation}
    \varrho^{AB} = \sum_{k} q_{k}\varrho^{A}_{k}\otimes\varrho^{B}_{k}
\end{equation}

with $\mathbf{1}^{T}\mathbf{q} = 1$. It is straightforward to check that expectation values no longer factorize for these states. These states are known as \emph{separable} states. Quantum correlations are, once again, a consequence of the Superposition Principle applied to composite systems: a (pure) quantum-correlated state doesn't admit a convex decomposition as in the previous expression, yet it still fulfils condition expressed in eq. \ref{dependence}. One example is the state used in the EPR paradox:

\begin{equation}
    \varrho^{AB}_{\psi} = |\psi\rangle\langle\psi| = \frac{1}{2}(|0\rangle|1\rangle - |1\rangle|0\rangle)(\langle 0|\langle 1| - \langle 1|\langle 0|)\neq\varrho^{A}\otimes\varrho^{B}
\end{equation}

Such states are called \emph{entangled}. Ascertain whether it was possible or not to describe entangled states in the context of a hidden variable theory was the task of Bell's inequalities.

The assumption of local realism entails the existence of joint probability distributions of a set of measurable quantities, regardless of whether they are observed or not. Any system will be at a definite state prior to being measured, which implies that the correlations between measurements at two different locations may depend on any (in general infinite) number of hidden variables $\lambda\in\Lambda$ (with $\Lambda$ a continuous set):

\begin{equation}
    C(i,j) = \langle\varrho^{AB}, M^{A}_{i}\otimes M^{B}_{j}\rangle = \int_{\Lambda} f(M^{A}_{i},\lambda)f(M^{B}_{j},\lambda)\mathbf{p}(\lambda)d\lambda
\end{equation}

where $\mathbf{p}(\lambda)$ is a probability distribution, and $f(M^{A}_{i},\lambda)$ is the probability of measuring outcome $i$ in system A, when the unknown hidden parameter is $\lambda$.

In an experiment proposed by Clauser, Horne, Shimony and Holt (CHSH) \cite{clauser1969pet} and carried out by Aspect and coworkers \cite{aspect1982ere}, it was possible to test whether entangled states admit a hidden-variable model (the CHSH inequality applies to two-states systems, inequalities for general systems have also been found. As an interesting case see \cite{braunstein1988itb}). Consider the measurements $\mathbf{M}^{A}_{1}$, $\mathbf{M}^{A}_{2}$, $\mathbf{M}^{B}_{1}(\theta)$, and $\mathbf{M}^{B}_{2}(\theta)$. These measurements have a probability overlap which depends on the angle $\theta$ between the two different setups (see fig. \ref{fig9}). They derived the following inequality, which holds for every theory that incorporates local realism:

\begin{equation}
    -2 \leq\langle \rho^{AB} , M^{A}_{1,i}\otimes(M^{B}_{1,j}(\theta) + M^{B}_{2,j}(\theta)) + M^{A}_{2,i}(\theta)\otimes(M^{B}_{1,j}(\theta) - M^{B}_{2,j}(\theta))\rangle\leq 2, \forall i, j
\end{equation}

\begin{figure}[h]
\centering
  \includegraphics[scale=0.3]{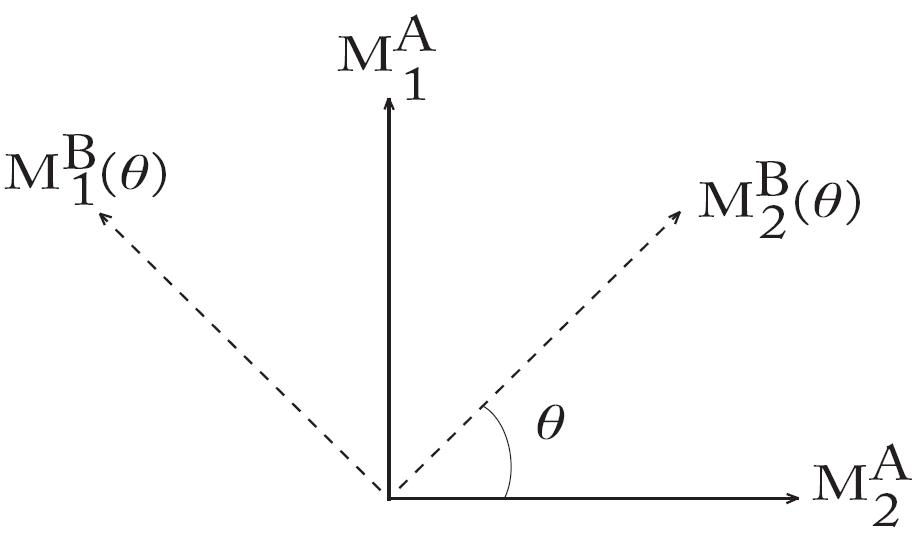}\\
  \caption{Violation of Bell's inequalities}\label{fig9}
\end{figure}

Quantum Mechanics predicts a violation of this inequality for some states. For the state $\varrho^{AB}_{\psi}$, the violation is maximum for $\theta = \frac{\pi}{4}, \frac{3\pi}{4}$, where we have that $|2\sqrt{2}|\nleq 2$. The conclusion is that EPR were in the right path: Quantum Mechanics is non-local.

\section{Information-Theoretic Considerations}

Von Neumann's Entropy is a generalization of Shannon's Entropy. It is zero for pure states, i.e., rank one density operators. In the case of bipartite systems, although the global system may be known with certainty to be in a pure state, such as for $\varrho^{AB}_{\psi}$, its marginal states, described by the reduced density operators $\varrho^{A} =Tr_{B}\varrho^{AB}_{\psi}$ and $\varrho^{B} = Tr_{A}\varrho^{AB}_{\psi}$, can be in a mixed state, so that their von Neumann entropy will be nonzero. This, once again, suggest that the whole system contains more information than the mere sum of the information contained in it parts.

As we will see, a pure state is entangled if and only if the von Neumann's entropy of any of its reduced density operators is nonzero.

One consequence of the above said is that conditional quantum information can be negative \cite{horodecki2005qic}. For pure entangled quantum states we have that $S(\varrho^{AB} = 0)$, so that:

\begin{equation}
    S(\varrho^{AB}) = S(\varrho^{A}) + S(\varrho^{B}|A)
\end{equation}

\begin{equation}
    S(\varrho^{B}|A) =  - S(\varrho^{A}) \leq 0
\end{equation}

This ``negative conditional information", with no counterpart in classical Information Theory, can be given an operational meaningful interpretation. If $S(\varrho^{B}|A)$ is negative, then A can reproduce the whole state $\varrho^{AB}$ just by means of classical communication, which is equivalent to say that quantum information can be transferred from B to A using only classical bits \cite{devetak2006omq}. Depending on its sign, conditional quantum information is the rate at which entanglement is created or consumed while transferring the state of be to state in A, and it is related to the quantum capacity of a channel.

\section{Unexpected Applications}

Given a total state that comprises locations A and B, an observer at location A can, by performing a local POVM on its subsystem, exert an influence upon the outcomes of observations at B (with A and B being spatially separated). This fact was named \emph{quantum steering} by Schr\"{o}dinger \cite{schrodinger1936prb}\cite{verstraete2002seq}. It was shown \cite{hughston1993ccq} that a local POVM at A can induce any ensemble $\{p_{k},\varrho_{k}\}$ at B provided that the reduced density operator at B admits a convex decomposition of the form $\varrho^{B} = \sum_{k}p_{k}\varrho_{k}$. If otherwise one could actually change the marginal statistics, superluminical communication would be achieved.

\subsection{Quantum Teleportation}

One of the brightest consequences of quantum steering is \emph{quantum teleportation.} Provided that observers at A and B share an EPR pair $\varrho^{AB}_{\psi}$, it is possible for them to teleport an unknown qubit\footnote{In general any qudit can be teleported. In this subsection and next one we will follow the original path of its discoverers and use qubits} in a pure state $\varphi^{C}$, initially at location A, to B using just local operations and classical communication (LOCC).

\begin{figure}[h]
\centering
  \includegraphics[scale=0.3]{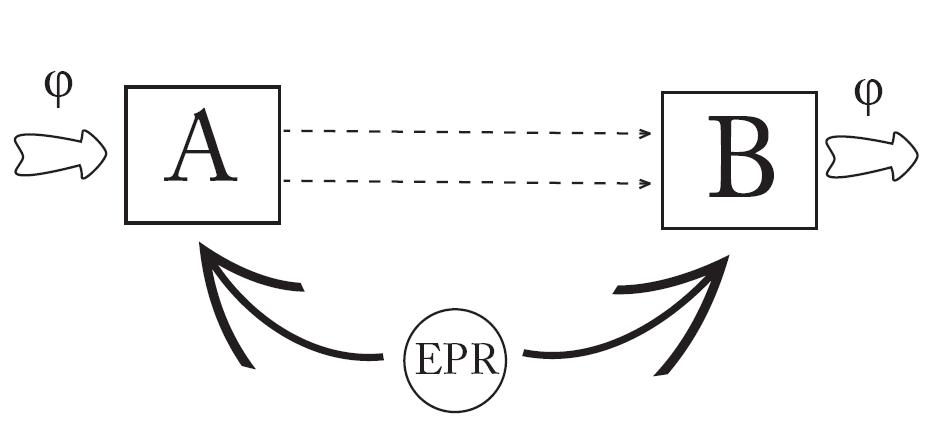}\\
  \caption{Quantum Teleportation: A quantum system in state $\varphi$ disappears at location A, and after some classical information has been sent from A to B, the system $\varphi$ appears at location B. A and B are causally disconnected. Despite its name, it is a rather prosaic effect, since it involves a protocol prescribed beforehand and requires that A and B share an EPR pair.}\label{fig10}
\end{figure}

Initially, the global system will be described by the state $\varrho^{AB}_{\psi}\otimes\varphi^{C}$. Observer at A
chooses a projective measurement $\mathbf{M}= \{M_{i}\}^{4}_{1}\in\mathfrak{M}^{A}\otimes\mathfrak{M}^{C}$ whose components are rank one operators such that:

\begin{equation}
    \langle M_{i}, \varrho^{AC}_{\psi_{j}}\rangle = \delta_{ij}
\end{equation}

where $\varrho^{AC}_{\psi_{j}} = |\psi_{j}\rangle\langle \psi_{j}|$ is any of the four states forming a \emph{Bell's basis}:

$$|\psi_{1}\rangle = \frac{1}{\sqrt{2}}(|0\rangle^{A}\otimes|0\rangle^{C} + |1\rangle^{A}\otimes|1\rangle^{C})$$
$$|\psi_{2}\rangle = \frac{1}{\sqrt{2}}(|0\rangle^{A}\otimes|0\rangle^{C} - |1\rangle^{A}\otimes|1\rangle^{C})$$
$$|\psi_{3}\rangle = \frac{1}{\sqrt{2}}(|0\rangle^{A}\otimes|1\rangle^{C} + |1\rangle^{A}\otimes|0\rangle^{C})$$
$$|\psi_{4}\rangle = \frac{1}{\sqrt{2}}(|0\rangle^{A}\otimes|1\rangle^{C} - |1\rangle^{A}\otimes|0\rangle^{C})$$

The EPR pair corresponds to the fourth vector of this basis. Once this basis has been selected, as long as the unknown system is in a pure state, it is possible to rewrite the initial state in the form \cite{bennett1993tuq}:

\begin{equation}
    \varrho^{AB}_{\psi}\otimes\varphi^{C} = \frac{1}{4}[\varrho^{AC}_{\psi_{1}}\otimes\varphi^{B} + \varrho^{AC}_{\psi_{2}}\otimes R^{\pi}_{z}(\varphi^{B}) + \varrho^{AC}_{\psi_{3}}\otimes R^{\pi}_{x}(\varphi^{B}) + \varrho^{AC}_{\psi_{4}}\otimes R^{\pi}_{y}(\varphi^{B})]
\end{equation}

Here $R^{\pi}_{k}(\varphi^{B})$ denotes a rotation of the state $\varphi^{B}$ of 180 degrees around axis k.

To teleport the unknown system, A performs the projective measurement in both systems A and C, so that its outcome determines with certainty in which state will the system in B is. Now, all A has to do is to encode its outcome in two bits and send them over to B. Once B receives the information, it will be possible to rotate the state back in the direction determined by the two bits. With 100\% accuracy, the initial unknown state will be obtained at a spatially separated location.

For mixed states, and for non-orthogonal POVM, it is still possible to teleport a system, but the process will be necessarily less efficient.

\subsection{Quantum Superdense Coding}\label{sec:superdense}

Another outstanding feat of quantum information is \emph{superdense coding} \cite{bennett1992cvo}, by which A can send to B two bits of information encoded in a single qubit\footnote{In fact, $2\log n$ bits can be sent using n-states systems}.

\begin{figure}[h]
\centering
  \includegraphics[scale=0.3]{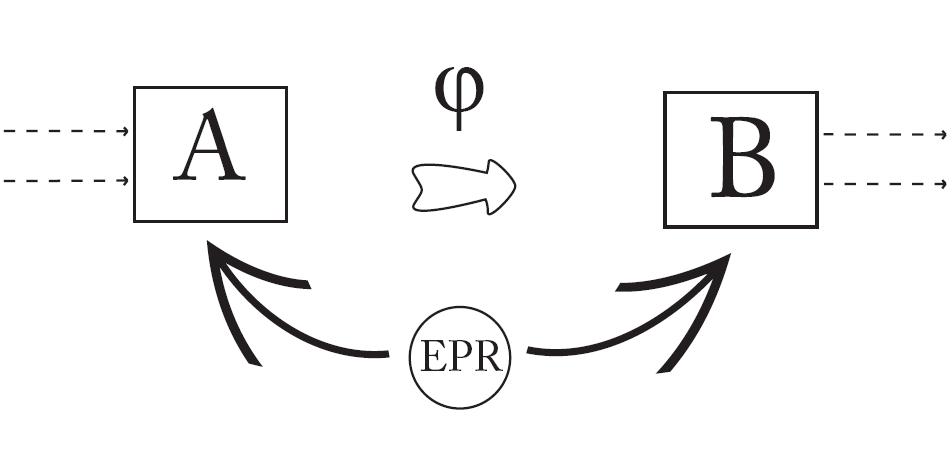}\\
  \caption{Superdense Coding: Provided that A and B share an EPR pair, A can change the global EPR state acting locally on its subsystem. Then A sends its part $\varphi =Tr_{B}\varrho^{AB}_{\psi_{i}}$ of the EPR state to B, who, measuring on a joint basis of the two subsystems, can extract two bits of information.}\label{fig11}
\end{figure}

Suppose that A and B share an EPR pair $\varrho^{AB}_{\psi_{4}}$. Observer at A causes its subsystem to evolve into one of the four possible states:

\begin{equation}
    (\mathbb{U}^{A}_{i}\otimes\mathbb{I}^{B})(\varrho^{AB}_{\psi_{4}})(\mathbb{U}^{A}_{i}\otimes\mathbb{I}^{B})^{\dagger} \Longrightarrow \varrho^{AB}_{\psi_{i}}
\end{equation}

The four possible operations that A can apply to the joint system are:

$$\mathbb{U}^{A}_{1} = \sigma^{A}_{y}$$
$$\mathbb{U}^{A}_{2} = \sigma^{A}_{x}$$
$$\mathbb{U}^{A}_{3} = \sigma^{A}_{z}$$
$$\mathbb{U}^{A}_{4} = \mathbf{1}^{A}$$

where the $\sigma$'s are Pauli matrices. After the manipulation, the system initially at A is sent to B. Observer at B chooses the projective measurement $\mathbf{M}= \{M_{i}\}^{4}_{1}\in\mathfrak{M}^{A}\otimes\mathfrak{M}^{B}$ with $\langle M_{i}, \varrho^{AB}_{\psi_{j}}\rangle = \delta_{ij}$. Thus B will gain two bits of information, while it only received one qubit.

Note that, whereas in quantum teleportation an unknown state was measured in A and the outcomes were encoded in classical bits, in superdense coding, a the ``future outcomes" of a measurement in B are encoded in a known qubit, which is later sent from A to B. First scenario is related to the capacity of transmitting quantum information through a classical channel, and second one is related to the capacity of transmitting classical information trough a quantum channel. Underlying these two processes lies the same phenomenon: Quantum Non-locality.

\chapter{Entanglement Theory}

Entanglement is a new resource for communication that lies at the very heart of lies at the very heart of Quantum Information
Theory. Thus, a theory of Entanglement which offers a qualitative description as well as quantitative measures is highly desirable.

Mainly, two difficulties surround this task. First one is to find a meaningful measure of the entanglement contained in a state. One way to obtain a suitable measure is to define beforehand some desirable properties that it should have:

\begin{description}
  \item[Scope] An entanglement measure is a map from the set of composite density operators to the positive real line:

      \begin{equation}
        E(\varrho^{AB})\in \mathbf{R}^{+}
      \end{equation}

  \item[Normalization] It should vanish only for separable states, and should be maximum for maximally entangled states:

      \begin{equation}
        0 \leq E(\varrho^{AB})\leq E(\varrho^{AB}_{\psi_{i}})
      \end{equation}

  \item[Monotonicity] $E(\varrho^{AB})$ should not increase under transformations involving only local operations and classical communication (LOCC). Consider the set of all LOCC transformations $\mathcal{T}_{LOCC}$. For any transformation $T\in\mathcal{T}_{LOCC}$:

      \begin{equation}
        E(T(\varrho^{AB}))\leq E(\varrho^{AB})
      \end{equation}

      As a consequence, one can derive the requirement of invariance under local unitary evolution. Let $\tilde{\varrho}^{AB} = (U\otimes V)\varrho^{AB}(U\otimes V)^{\dagger} $, for any two unitary operators, then:

       \begin{equation}
        E((U\otimes V)\varrho^{AB}(U\otimes V)^{\dagger}) \leq E(\varrho^{AB})
       \end{equation}

       \begin{equation}
        E((U\otimes V)^{\dagger}\tilde{\varrho}^{AB}(U\otimes V)) \leq E(\tilde{\varrho}^{AB})
       \end{equation}

      whence we obtain $E(\varrho^{AB}) = E((U\otimes V)\varrho^{AB}(U\otimes V)^{\dagger})$.

  \item[Convexity] Such a desirable property arises naturally from the reasonable assumption that mixing two states should not increase the entanglement contained in them:

        \begin{equation}
              E(\lambda \varrho^{AB} + (1 -\lambda)\varphi^{AB}))\leq \lambda E(\varrho^{AB}) + (1 - \lambda) E(\varphi^{AB})
        \end{equation}

  with $0\leq\lambda\leq1$.

  \item[Continuity] Intuitively, if $\varrho^{AB}$ is slightly perturbed into $\varphi^{AB}$, the subsequent change in the entanglement measure should be small. This is expressed as:

      \begin{equation}
        \lim_{\|\varrho - \varphi\|\rightarrow 0}E(\varrho^{AB}) - E(\varphi^{AB}) = 0
      \end{equation}

  \item[Subadditivity] The communication tasks that one is able to perform in the possession of several entangled pairs shouldn't be more than the sum of those permitted individually by each pair:

      \begin{equation}
        E(\varrho^{AB}\otimes\varphi^{AB}) \leq E(\varrho^{AB}) + E(\varphi^{AB})
      \end{equation}

    for the case when one has many copies of the same state, the demand which is often encountered is thatof weak subadditivity:

    \begin{equation}
        \frac{E(\varphi^{\otimes N})}{N} = E(\varphi)
    \end{equation}

\end{description}

The second difficulty is how to compute an entanglement measure for any given state, which as we will see, is a far from trivial task.

\section{Pure States}

For pure states, a satisfactory theory exists and there are procedures both to detect and quantify the entanglement of a given state. The von Neumann's entropy of any of the reduced density operators satisfies the requirements exposed above \cite{donald2002ute}\cite{popescu1997tam} and is directly related to the \emph{Schimdt's decomposition}\footnote{The Schmidt's decomposition of a bipartite state is the projection of the state onto an orthonormal product basis of the two Hilbert spaces. If a state is separable, its vector state will be pointing parallel to one of the vectors of this basis. If it is entangled, its vector state will have more than one component.} of the vector state. For a maximally entangled symmetric state in dimension $N$, its Schmidt's decomposition is:

\begin{equation}
    |\psi_{1}\rangle^{AB} =  \frac{1}{\sqrt{N}}\sum^{N - 1}_{i=0}|i\rangle^{A}\otimes|i\rangle^{B}
\end{equation}

and the von Neumann's entropy of its reduced density operator attains its maximum:

\begin{equation}
    S(Tr_{B}\varrho^{AB}_{\psi_{1}}) = \log N
\end{equation}

If the state is separable, it will necessarily consist of two pure states at each location, so the entropy of any of the local density operators will be zero.

\subsection{Entanglement Distillation}

The preference of von Neumann's entropy of the reduced density operator over other candidates relies also on another reason: it quantifies the amount of maximally entangled states that one can obtain from an arbitrary large number of arbitrary density operators, by some LOCC transformation \cite{bennett1996cpe}.

\begin{thm}\label{distillation}\textcolor[rgb]{0.00,0.00,1.00}{[Entanglement Distillation]} Given $m$ identical copies of one arbitrarily entangled pure state, $(\varphi^{AB})^{\otimes m}$, then there exist a LOCC scheme $T\in\mathcal{T}_{LOCC}$ such that it is possible to obtain $n \leq m$ copies of a maximally entangled state:

\begin{equation}
    \lim_{m\rightarrow\infty}\| T((\varphi^{AB})^{\otimes m}) - (\varrho^{AB}_{\psi_{1}})^{\otimes n}\| = 0
\end{equation}

and the rate at which this can be done is given by the von Neumann's entropy of any of the reduced density operators of the initial state:

\begin{equation}
    R = \lim_{m\rightarrow\infty}\frac{n}{m} = S(Tr_{B}\varphi^{AB})
\end{equation}

\end{thm}

Our proof for this theorem will need typicality arguments, so it will be given in next chapter.

The importance of entanglement distillation is that most communication tasks rely on maximally entangled states in order to yield acceptable transmission fidelities(see teleportation and superdense coding). Due to the stochastic influence of channels on quantum information, it is very difficult to prevent a transmitted entangled particle from being corrupted with some noise, and protocols must be devised to restore the initial entanglement, at the expense of sending more particles.

The converse procedure is called \emph{entanglement dilution}, by which $n$ copies of the maximally entangled state can be used to obtain $m\geq n$ copies of an arbitrarily entangled pure state. However it doesn't seem to be as practical as entanglement distillation.

\section{Mixed States}

For the general case of mixed states, the theory of entanglement is far from complete. There are several several measures, based on inequivalent criteria. The prevalence of any of this candidates has not yet occurred. Our approach to the study of entanglement focuses on two distinct areas, detection and quantification of entanglement. These two concepts are tightly interrelated, and this scheme is just a matter of taste.

It was shown that the \emph{separability problem}, i.e. to ascertain whether a given state is separable or entangled, belongs to the class of NP-hard problems \cite{gurvits2003cdc}. This means that, on a realistic basis, we should not expect to measure (and detect) entanglement with arbitrary accuracy. Several relaxations of the separability problem have been proposed in the context of \emph{convex programming} \cite{doherty2002dsa}\cite{doherty2004cfs}\cite{brandao2004smm}\cite{branduao2004rsp}. Here we will study more in depth the approach suggested in \cite{herreramarti2008sec}, which a offers a geometric intuition of the space of composite density operators.

For a thorough review of this concepts readers might check \cite{terhal2001dqe}\cite{bruss2001ce}.

\subsection{Detection}

Several criteria have been proposed to check whether a state exhibits quantum or just classical correlations. Here we shall list some of them:

\begin{description}
  \item[Peres-Horodecki Criterion] Peres showed that \emph{Positivity under Partial Transposition} (PPT) of the density operator was a necessary condition for separability \cite{peres1996scd}. A bipartite state $\varrho^{AB}$ is separable if:

      \begin{equation}
        (\varrho^{AB})^{T_{B}} = \sum_{k} \varrho^{A}_{k}\otimes (\varrho^{B}_{k})^{T} \geq 0
      \end{equation}

      that is, if it remains positive semidefinite under transposition of just one local density operators. The Horodecki traced back this argument to the theory of positive maps \cite{horodecki1996sms}, and demonstrated that for systems of dimension $2\times2$ (two entangled qubits)and $2\times3$ (one qubit entangled with one qutrit) this criterion is also a sufficient condition. Why the separability problem is solved, despite being NP-hard in general, finds an explanation in the fact that, for $2\times2$ and $2\times3$ systems, any positive map is of the form $\mathbb{L} = \mathbb{S}_{1} + \mathbb{S}_{2}\circ\mathbb{T}$ ($\mathbb{S}_{1,2}$ are completely positive maps and $\mathbb{T}$ is the transposition map), so no further search is needed to fully characterize separable states.\footnote{Positive maps are those maps for which $\mathbb{L}(\varrho)\geq 0, \forall\varrho$. A positive map $\mathbb{L}$ is completely positive if and only if $(\mathbb{I}_{n}\otimes\mathbb{L})(\varrho)\geq 0, \forall\varrho,n$. For classification, only positive maps are interesting. It is easy to see that the PPT criterion relies on a positive map.}

  \item[Majorization] Another necessary condition for separability is the \emph{Majorization Criterion}, which although being less effective in detecting entanglement than the PPT criterion, reveals a thermodynamic aspect of non-locality. Consider the composite system $\varrho^{AB}$ and its reduced density operator $\varrho^{A}$. Denote by $\lambda(\varphi)$ the vector of the eigenvalues of $\varphi$, arranged in decreasing order.$\varrho^{AB}$ is separable if:

      \begin{equation}
        \lambda(\varrho^{AB})\prec\lambda(\varrho^{A})
      \end{equation}

      where $\prec$ is a pre-order relation meaning that the vector $\lambda(\varrho^{A}) - \lambda(\varrho^{AB})$ lies inside some \emph{positive cone} so that $S(\varrho^{AB})\geq S(\varrho^{A})$ \cite{nielsen1999cce}\cite{nielsen2001ssm}. This, in turn, implies that if the state is entangled, then $S(\varrho^{AB})\leq S(\varrho^{A})$, so once again, we see that Quantum Mechanics allows information to be stored in a composite system in a holistic manner, regardless of its parts.

  \item[Entanglement Witnesses] In \cite{horodecki1996sms}, the Horodecki introduced Entanglement Witnesses (EW). An EW is a Hermitian operator $W = W^{\dagger}$ such that:

      \begin{eqnarray}
        \langle W, \sigma \rangle  &\geq& 0 \text{, for all separable $\sigma$ } \\
        \langle W, \varrho\rangle &<& 0 \text{, for some entangled $\varrho$}
      \end{eqnarray}

       which is a consequence of the Hahn-Banach Theorem \cite{alt2002lfa} in functional analysis, since EW are directly related to positive maps \cite{jamiolkowski1972ltp} defined on a Banach space. This is valid for algebras of arbitrary dimension, so it will prove to be a very useful concept.

  \end{description}

As stated before, there is a duality between positive maps and operators. In low dimensions ($2\times2$ and $2\times3$), any positive map admits the decomposition $\mathbb{L} = \mathbb{S}_{1} + \mathbb{S}_{2}\circ\mathbb{T}$, and any EW can be written as \cite{jamiolkowski1972ltp}:

\begin{equation}\label{decomposable}
    W = (\mathbb{I}\otimes\mathbb{L})(\varrho^{AB}_{\psi_{1}}) = P + (\mathbb{I}\otimes\mathbb{T})Q
\end{equation}

Where P and Q are nonnegative operators. These EW receive the name of \emph{decomposable EW}. For $P=0$ and $Q=I$ one gets the PPT criterion. Nevertheless, for higher dimensions there exist EW which are not of the form \ref{decomposable}. A consequence is that there will be entangled states for which $\langle W, \varrho\rangle \geq 0$. These states are called \emph{Positive Partial Transposed Entangled States}(PPTES), and it was shown in \cite{horodecki1998mse}, that this kind of entangled states cannot be used in distillation procedures. For this reason, the entanglement contained in them is called \emph{Bound Entanglement}, since it cannot be extracted for communication tasks.

Other criteria to detect and quantify Bound Entanglement have been proposed, such as non-decomposable EW \cite{lewenstein2000oew}, Schmidt number Witnesses \cite{sanpera2001snw}, Robust Semidefinite Programming \cite{branduao2004rsp} and, more recently, a geometric approach based on separating hyperplanes \cite{bertlmann2008gew}. We will pursue this geometric interpretation of entanglement in next section.

\subsection{Quantification}

There exist several candidates, depending on which criterion one takes as more fundamental. Some of them have operational definitions and some have not.

\begin{description}

  \item[Entanglement cost, $E_{C}$] It quantifies how many maximally entangled pairs are needed to generate a given entangled state, minimized over all possible dilution protocols:

      \begin{equation}
        E_{C}(\varrho) = \min_{\mathcal{T}_{LOCC}} \lim_{m\rightarrow\infty} \frac{n}{m}
      \end{equation}

      where $n\leq m$ is the number of maximally entangled pairs, $\varrho^{\otimes n}_{\psi_{1}}$, whose entanglement is diluted into m copies of the original state, $\varrho^{\otimes m}$.

  \item[Distillable entanglement, $E_{D}$] It is a measure of how many maximally entangled pairs can be obtained by performing an optimal distillation protocol to an asymptotic number m of copies of the given state:

     \begin{equation}
         E_{D}(\varrho) = \max_{\mathcal{T}_{LOCC}} \lim_{m\rightarrow\infty} \frac{n}{m}
     \end{equation}

      with n and m as before.

  \item[Relative entropy of entanglement, $E_{R}$] Analogously to its classical counterpart, $E_{R}$ (\cite{vedral1997ema}) can be thought of as a measure for the extent that one can confuse two probability distribution, result known as Sanov's Theorem (see \cite{cover2006eit}). But this case, it quantifies to which amount an entangled state can be taken as separable. The \emph{relative entropy of entanglement} is:

      \begin{equation}
        E_{R}=\min_{\sigma\in\mathcal{S}}S(\varrho||\sigma) = \min_{\sigma\in\mathcal{S}} Tr(\varrho\log\varrho - \varrho\log\sigma)
      \end{equation}

  \item[Entanglement of formation, $E_{F}$] Any density operator has a non-unique convex decomposition of the form $\varrho = \sum_{k} p_{k} \varrho_{k}$, where $\varrho_{k}$ are rank one density operators. Its \emph{entanglement of formation} is the averaged von Neumann's entropy of those pure states, minimized over all possible convex decompositions:

      \begin{equation}
        E_{F} = \min_{\{p_{k},\varrho_{k}\}} \sum_{k} p_{k} S(\varrho_{k})
      \end{equation}

\end{description}

As for the detection case, it is not known whether this measures are equivalent. The values of this quantities are only known for some cases. In the $2\times2$ case, the entanglement of formation can be exactly computed thanks to a measure known as \emph{concurrence}\cite{wootters1998efa}. For any entanglement measure $E(\varrho)$, $E_{D}(\varrho)\leq E(\varrho)\leq E_{C}(\varrho)$. For bound entangled states, this is trivially satisfied.

\section{Geometric Insights into Entanglement}

The set of all density operators is a convex set, which follows from probability arguments. Mixing cannot increase entanglement, hence the set of separable density operators is also convex. We will denote this two sets by $\mathcal{D}$ and $\mathcal{S}$, respectively.

\begin{figure}[h]
\centering
  \includegraphics[scale=0.4]{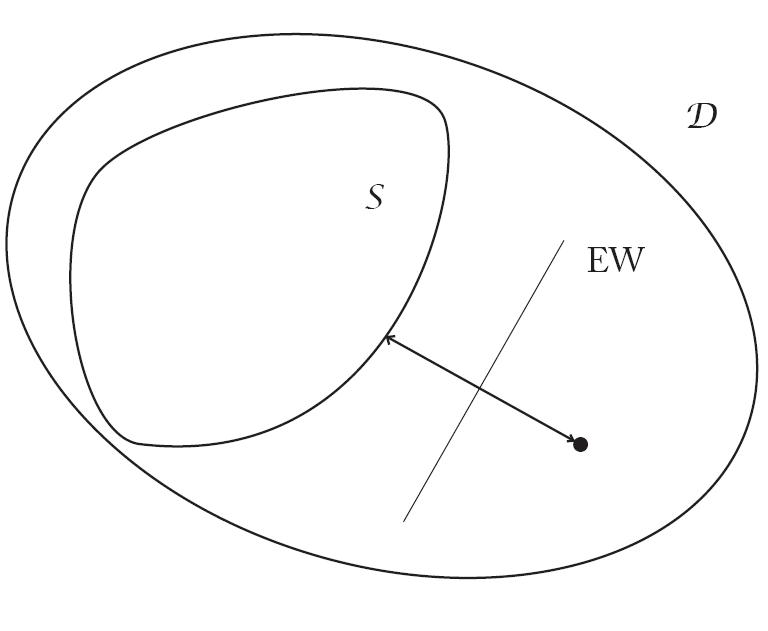}\\
  \caption{Set of composite states $\mathcal{D}$. $\mathcal{S}$ is a convex subset containing all states that exhibit classical correlations. $\mathcal{D}\setminus\mathcal{S}$ contains all states exhibiting quantum correlations. For a given density operator in $\mathcal{D}$, its distance to the separable set is a measure of the entanglement that it contains. An Entanglement Witness will separate $\mathcal{S}$ from a convex subset of $\mathcal{D}\setminus\mathcal{S}$.}\label{fig12}
\end{figure}

It is easy to see that entanglement witnesses constitute hyperplanes which split $\mathcal{D}$ into two subsets, one of which strictly contains $\mathcal{S}$. Let $W=W^{\dagger}$ be an EW, then:

\begin{equation}\label{rho}
    \langle W, \varrho   \rangle < 0 \text{, for some $\varrho \in \mathcal{D}\setminus\mathcal{S}$}
\end{equation}

\begin{equation}\label{sigma}
    \langle W, \sigma \rangle \geq 0 \text{, $\forall  \sigma \in \mathcal{S}$}
\end{equation}

Clearly this defines a hyperplane dividing $\mathcal{D}$: $\langle W, \sigma \rangle \geq \langle W, \varrho \rangle$. In our notation, an EW is optimal if (\ref{rho}) holds for the largest number of $\varrho$'s. Intuitively, such an EW will be tangent to $\mathcal{S}$ (see fig \ref{fig12}). We will illustrate this fact in a moment.

A geometric measure of  the entanglement contained in a state is its distance to the set of separable density operators. The distance of a density operator $\varrho$ to the separable set $\mathcal{S}$ is:

\begin{equation}\label{distance}
    D = \min_{\sigma\in\mathcal{S}}\|\varrho - \sigma\|
\end{equation}

\subsection{Duality between Detection and Quantification}

It is a general result from geometric optimization that the problem of finding a separating hyperplane between a point $\mathbf{p}$ and  a convex set $\mathcal{C}$ is dual to the problem of finding the distance between $\mathcal{C}$ and $\mathbf{p}$ \cite{boyd2004co}. In matrix space language, this duality can be illustrated as follows. The problem (ref{distance}) can be expressed as:

\begin{eqnarray}\label{distance2}
    \min &\| \tau \|\\
    \text{such that}&\varrho - \sigma =\tau \nonumber \\
    &\sigma \in \mathcal{S}\nonumber
\end{eqnarray}

with variables $\tau$ and $\sigma$. The Lagrangian of (\ref{distance2}) is:

\begin{equation}\label{lagrangian}
    \mathcal{L} = \| \tau \| + \langle W,\varrho - \sigma - \tau\rangle
\end{equation}

where W is the Lagrange multiplicator associated to the equality constraint. It is not hard to see that it represents a a hyperplane. Noting that $\langle W,\tau \rangle \leq \| W \|\| \tau \|$, the dual function can be written as:

\begin{equation}\label{dualfunction}
    g(W) = \min_{\sigma \in \mathcal{S},\tau} [\langle W,\varrho \rangle - \langle W,\sigma \rangle + \| \tau \|(1-\| W \| + \delta)]
\end{equation}

where the parameter $\delta\geq 0$ is related to the relative orientation between the hyperplane represented by W, and the line going from the separable set $\mathcal{S}$ to the density operator $\varrho$, and it is equal to zero if and only if they are perpendicular. For (\ref{dualfunction}) to be bounded from below in $\tau$, the additional constraint $\| W \| - \delta \leq 1$ must be included. So the dual problem of (\ref{distance2}) is:

\begin{equation}\label{dualproblem}
  \max_{\| W \| - \delta \leq 1}[ \min_{\sigma \in \mathcal{S}} [\langle W,\varrho \rangle - \langle W,\sigma \rangle]]
\end{equation}

It is straightforward to check that the optimal value of (\ref{dualproblem}) is attained if and only if W is an optimal EW. This result, also known as the Bertlmann-Narnhoffer-Thirring Theorem (see Ref. \cite{bertlmann2002gpe}), will let us trace a link between the entanglement detection and quantification problems (compare also with Refs. \cite{brandao2005qew}\cite{eisert2007qew})

\subsection{Ellipsoidal Classification}

The basic premise of this method is that the set of separable states $\mathcal{S}$ can be approximated by a \emph{Minimum Volume Covering Ellipsoid} (MVCE) of an ensemble of vectors corresponding to some separable density operators. Then, the following classification scheme can be adopted: if a vector falls inside the MVCE, it will be taken as separable, and if it falls outside, it will be regarded as entangled.

An ellipsoid centered at $\mathbf{x}_{c}$ can be expressed as:

\begin{equation}\label{ellipsoid}
    \mathcal{E}= \{ \mathbf{x} | (\mathbf{x} - \mathbf{x}_{c})^{T}A(\mathbf{x} - \mathbf{x}_{c}) \leq 1 \}
\end{equation}

where $A = A^{T}$ is a positive definite matrix of dimension $N^{2}- 1$. The volume of this ellipsoid is proportional to $\det(A^{-1/2})$.

Since in matrix space quadratic forms are not defined, one needs to work in a real vector space to build this ellipsoid. We first obtain an ensemble of ``separable vectors" by means of tensorially multiplying states along all directions specified by some canonical basis. For instance, in the $2\times 2$ case, this ensemble would be $\{\mathbf{x}^{sep}_{i}\} =\{(1,0,0)\otimes(1,0,0), (1,0,0)\otimes(-1,0,0),(1,0,0)\otimes(0,1,0), (1,0,0)\otimes(0,-1,0), ...,(0,0,-1)\otimes(0,0,-1)\}$ . Later on we will see that it is convenient to vary the norm $\|\mathbf{x}^{sep}_{i}\|_{2}$ of these vectors. This procedure ensures that all vectors will lie as spaced as possible in the separable set $\mathcal{S}$. Secondly, we minimize the volume of an ellipsoid, constrained to have all generated ``separable vectors" falling inside it. One way to obtain the MVCE of this ensemble would be to solve the following problem:

\begin{eqnarray}\label{mvce}
    \min &\log\det A^{-1}\\
    \text{such that}&(\mathbf{x}^{sep}_{i}-\mathbf{x}_{c})^{T}A(\mathbf{x}^{sep}_{i}-\mathbf{x}_{c})\leq 1 \nonumber
\end{eqnarray}

with variables A and $\mathbf{x}_{c}$. Here, logarithm was taken in order to drop off proportionality terms. Despite the exponential growth of the dimension of the associated vector space, interior point methods used for minimization still converge polynomially to a solution in dimension as large as 1000, or more \cite{boyd2004co}.

\begin{figure}[h]
\centering
  \includegraphics[width=7.5cm]{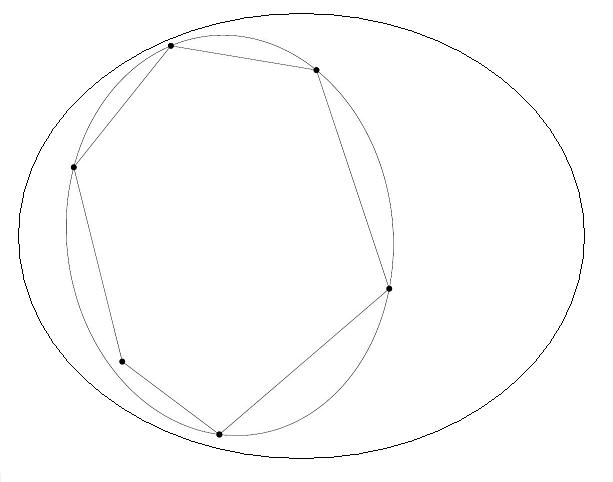}\\
  \caption{The vertices of the polytope are the generated ``separable vectors" of which the MVCE is found. The larger set corresponds to the whole space of density operators}\label{DOspace}
\end{figure}

\subsection{Results for $2\times 2$ and $2\times 3$ Systems}

The separability problem is solved for $2\times 2$ and $2\times 3$ systems, thanks to the PPT criterion. One can use this fact to benchmark the method. The original problem (\ref{distance}) casted as:

\begin{eqnarray}\label{distance3}
    \min &\|\varrho - \sigma\|_{F} \\
    \text{such that} &\sigma^{T_{A}}\geq 0 \nonumber
\end{eqnarray}

which gives the true results.  The second problem is to find the MVCE through (\ref{mvce}), and compute the distance to this ellipsoid in a similar way:

\begin{eqnarray}\label{distance4}
  \min &\|\mathbf{r} - \mathbf{s}\|_{2} \\
  \text{such that} &(\mathbf{s}-\mathbf{x}_{c})^{T}A(\mathbf{s}-\mathbf{x}_{c})\leq 1 \nonumber
\end{eqnarray}

where $\mathbf{r}$ and $\mathbf{s}$ stand for the vectorized counterparts of $\varrho$ and $\sigma$. The results obtained for pure vectors ($\|\mathbf{x}^{sep}_{i}\|_{2} = 1 $) are rather discouraging: whereas none of the generated ``separable vectors" fell outside the MVCE, only 12.7\% of the ``entangled vectors" are detected. However, the ellipsoid can be shrunk by reducing the norm of the generated ensemble $\{\mathbf{x}^{sep}_{i}\}$. At the expense of letting some ``separable vectors" fall outside the ellipsoid, the number of correctly classified ``entangled vectors" increases. The event that a true ``separable vector" falls outside the MVCE will be a \emph{false positive}, while if an ``entangled vector" falls inside the  MVCE, it will be \emph{false negative}. Stepwise reducing the norm of the vectors belonging to the separable ensemble Tables 1 and 2 are obtained.

\begin{table}
  \centering
\begin{tabular}{|c||c|c|}
  \hline
  \multicolumn{3}{|c|}{2 x 2 Systems}\\
  \hline
  Norm & False Positives & False Negatives \\
  \hline \hline
  0.1 & 962 & 0 \\
  0.2 & 868 & 0 \\
  0.3 & 687 & 0 \\
  0.4 & 484 & 0 \\
  0.5 & 287 & 15 \\
  0.6 & 180 & 184 \\
  0.7 & 92 & 410 \\
  0.8 & 32 & 600 \\
  0.9 & 5 & 755 \\
  1.0 & 0 & 873 \\
  \hline
\end{tabular}
\caption{Number of misclassified vectors in a sample of 1000 ``separable vectors" and 1000 ``entangled vectors", as a function of the Euclidean norm of the vectors of the separable ensemble}
\end{table}

\begin{table}
  \centering
\begin{tabular}{|c||c|c|}
  \hline
  \multicolumn{3}{|c|}{2 x 3 Systems}\\
  \hline
  Norm & False Positives & False Negatives \\
  \hline \hline
  0.1 & 949 & 0 \\
  0.2 & 812 & 0 \\
  0.3 & 597 & 0 \\
  0.4 & 427 & 52 \\
  0.5 & 269 & 196 \\
  0.6 & 160 & 389 \\
  0.7 & 80 &  572 \\
  0.8 & 34 & 699 \\
  0.9 & 11 & 807 \\
  1.0 & 0 & 900 \\
  \hline
\end{tabular}
\caption{Number of misclassified vectors in a sample of 1000 ``separable vectors" and 1000 ``entangled vectors", as a function of the Euclidean norm of the vectors of the separable ensemble}
\end{table}

\begin{figure}[h]
\centering
  \includegraphics[width=9cm]{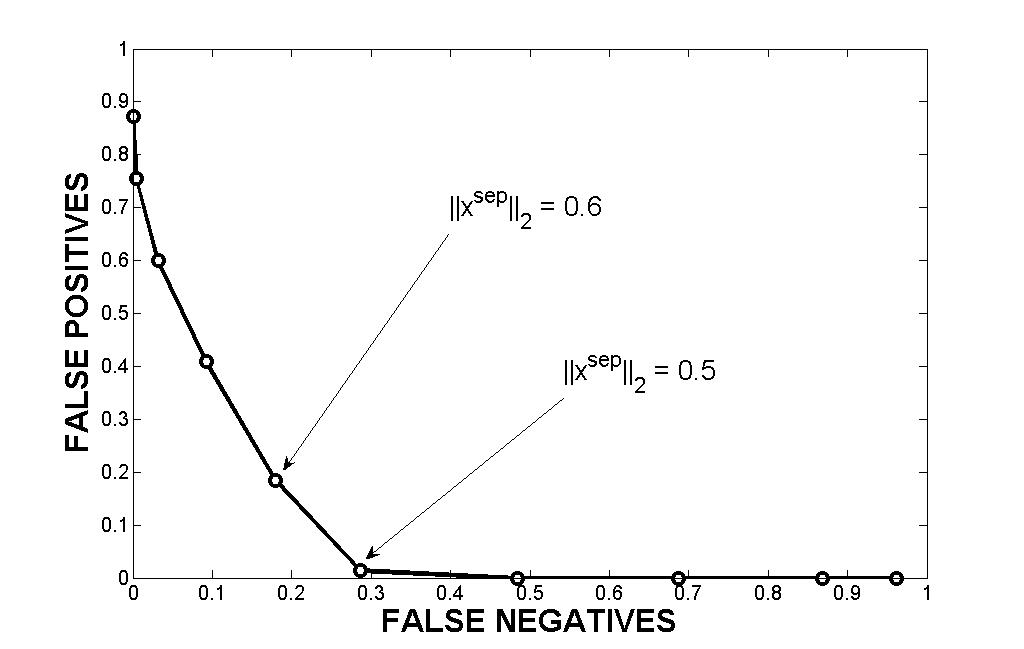}\\
  \caption{False Negatives versus False Positives for $2 \times 2$ systems, showing that there exists an area where the probability of wrongly classifying a vector can be brought down to 15.1\%, between $\|\mathbf{x}^{sep}_{i}\|_{2} = 0.6$ and $\|\mathbf{x}^{sep}_{i}\|_{2} = 0.5$}\label{fnfp22}
\end{figure}

\begin{figure}[h]
\centering`
  \includegraphics[width=9cm]{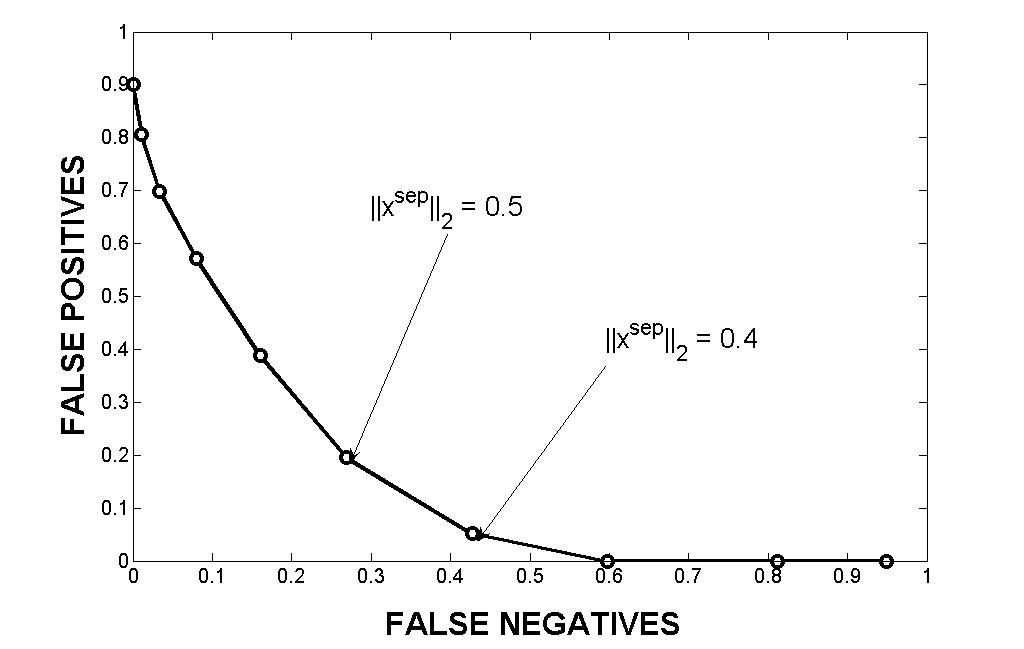}\\
  \caption{False Negatives versus False Positives for $2 \times 3$ systems. The error probability can be reduced to 23.2\%, between $\|\mathbf{x}^{sep}_{i}\|_{2} = 0.5$ and $\|\mathbf{x}^{sep}_{i}\|_{2} = 0.4$}\label{fnfp23}
\end{figure}

There is a trade-off between the number of correctly classified states and non-ambiguousness of the test. The relevant area of $2\times 2$ systems is between norms  $\|\mathbf{x}^{sep}_{i}\|_{2} = 0.6$ and $\|\mathbf{x}^{sep}_{i}\|_{2} = 0.5$, as can be seen in Fig. \ref{fnfp22}. A measure of entanglement ought to  be as unambiguous as possible, and thus the best choice is $\|\mathbf{x}^{sep}_{i}\|_{2} = 0.5$, since for this case only about 1.5\% of the ``entangled vectors" are misclassified. For this choice, in general, a vector will be misclassified 15.1\% of the time. For $2\times 3$ systems (see fig. \ref{fnfp23}), the MCVE approximates somewhat less efficiently the separable set. However, still 76.8\% of the vectors are correctly classified, in the area comprised between $\|\mathbf{x}^{sep}_{i}\|_{2} = 0.5$ and $\|\mathbf{x}^{sep}_{i}\|_{2} = 0.4$. In these systems, it misclassifies at least 5.2\% of the ``separable vectors".

\subsection{Pseudo-Entanglement Witnesses}

For a vector space endowed with the Euclidean norm, there is a simple way to construct a tangent hyperplane to a given ellipsoid. We can use this fact to build realistic observables amenable to a laboratory setting.

The tangent hyperplane to the ellipsoid can be expressed as:

\begin{equation}\label{tangent1}
    \nabla_{\mathbf{x}}[(\mathbf{x}-\mathbf{x}_{c})^{T}A(\mathbf{x}-\mathbf{x}_{c}) - 1]_{s_{0}}(\mathbf{r} - \mathbf{s}_{0}) = 0
\end{equation}

where $\mathbf{s}_{0} = P_{\mathcal{E}}(\mathbf{r})$ is the projection of the vectorized density operator under study onto the MVCE. It can be expressed in affine form as:

\begin{equation}\label{tangent2}
    (\mathbf{s}_{0} - \mathbf{x}_{c})^{T}{A(\mathbf{r} - \mathbf{x}_{c}) = 1}
\end{equation}

(compare with Ref. \cite{bertlmann2005oew}). It is important to keep in mind that, although the hyperplanes introduced in (\ref{tangent2}) very much resemble an Entanglement Witness, they are not so in general. This is because the MCVE may in general be a proper subset of the separable set $\mathcal{S}$, and no tangent hyperplane to this MVCE will strictly separate $\mathcal{S}$ from any entangled state. Nevertheless, these \emph{Pseudo}-EW can be used to estimate the amount of entanglement contained in a given entangled matrix $\varrho$ via (\ref{dualproblem}), which at the optimal value will be equal to (\ref{distance2})\cite{boyd2004co}. For an illustration of entanglement estimation see fig. \ref{bounds}.

\begin{figure}[h]
\centering
  \includegraphics[width=9cm]{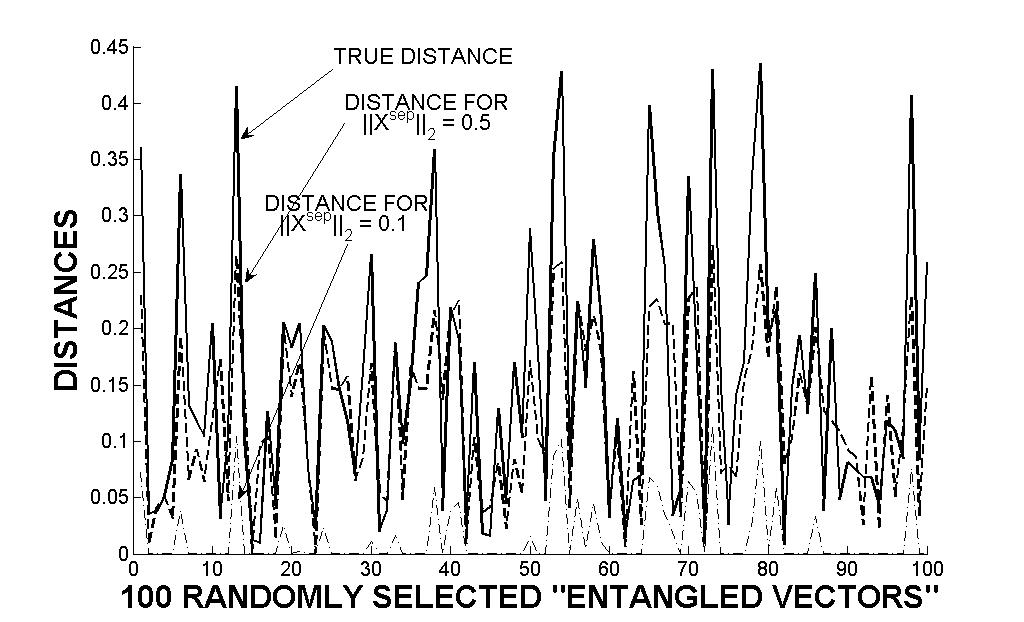}\\
  \caption{For 100 random ``entangled vectors" of a $2 \times 2$ system, the continuous black line is the true distance to the separable set $\mathcal{S}$, while the dashed line stands for the distance of the vectors for a MVCE of ``separable vectors" of norm $\|\mathbf{x}^{sep}_{i}\|_{2} = 0.5$. At the bottom, the pointed line represents the distances obtained for norm $\|\mathbf{x}^{sep}_{i}\|_{2} = 1$}\label{bounds}
\end{figure}

\subsection{Bound Entanglement Detection}

For composite systems of dimension higher than 6, there is a special kind of entangled states that cannot be used, in principle, to enhance communication. The entanglement contained in these states cannot be distilled to obtain pure entangled states \cite{horodecki1998mse}, and it receives the name of \emph{Bound Entanglement} (BE). The PPT criterion fails to detect this kind of states, and it just becomes a necessary condition for quantum correlations to arise. Other criteria to detect and quantify BE have been proposed, such as non-decomposable EW \cite{lewenstein2000oew}, Schmidt number Witnesses \cite{sanpera2001snw}, and, more recently, a geometric approach based on separating hyperplanes \cite{bertlmann2008gew}.

The MVCE approach is in the spirit of the latter of the aforementioned methods, but instead of hyperplanes, we shall use the MVCE in order to detect BE. Intuitively, the ellipsoid covering a set of ``separable vectors" should leave bound entangled states on its outside. This fact is studied in 3X3 systems, where a parametrization of bound entangled states, due to P. Horodecki, is available \cite{horodecki1997sca}. These states $\varrho_{BE}$ depend on a scalar $a \in [0,1]$, and are given by:

$$\varrho_{BE}(a) = \frac{1}{8a + 1}
\left(
  \begin{array}{ccccccccc}
    a & 0 & 0 & 0 & a & 0 & 0 & 0 & a \\
    0 & a & 0 & 0 & 0 & 0 & 0 & 0 & 0 \\
    0 & 0 & a & 0 & 0 & 0 & 0 & 0 & 0 \\
    0 & 0 & 0 & a & 0 & 0 & 0 & 0 & 0 \\
    a & 0 & 0 & 0 & a & 0 & 0 & 0 & a \\
    0 & 0 & 0 & 0 & 0 & a & 0 & 0 & 0 \\
    0 & 0 & 0 & 0 & 0 & 0 & \frac{1 + a}{2} & 0 & \frac{\sqrt{1 - a^{2}}}{2} \\
    0 & 0 & 0 & 0 & 0 & 0 & 0 & a & 0 \\
    a & 0 & 0 & 0 & a & 0 & \frac{\sqrt{1 - a^{2}}}{2} & 0 & \frac{1 + a}{2} \\
  \end{array}
\right)
$$

Surprisingly, for norms of the generated separable ensemble of 0.6 and below, all bound entangled states are detected. The obtained results are shown in Table 3.

\begin{table}
  \centering
\begin{tabular}{|c||c|}
  \hline
  \multicolumn{2}{|c|}{3 x 3 Systems}\\
  \hline
  Norm &  Detected States\\
  \hline \hline
  0.1 & 1000\\
  0.2 & 1000\\
  0.3 & 1000\\
  0.4 & 1000\\
  0.5 & 1000\\
  0.6 & 1000\\
  0.7 & 226\\
  0.8 & 149\\
  0.9 & 107\\
  1.0 & 79\\
  \hline
\end{tabular}
\caption{1000 bound entangled states were generated, with parameter ``a" running from 0.001 to 1. The distances of the associated vectors to the different MVCEs were obtained. For norms of the separable ensemble of 0.6 and below, all bound entangled states were detected}
\end{table}

As expected, the distance to the MVCE of the detected states linearly depends on the norm of the associated density operator. This interdependence is depicted in Fig. \ref{BEdist}

\begin{figure}[h]
\centering
  \includegraphics[width=9cm]{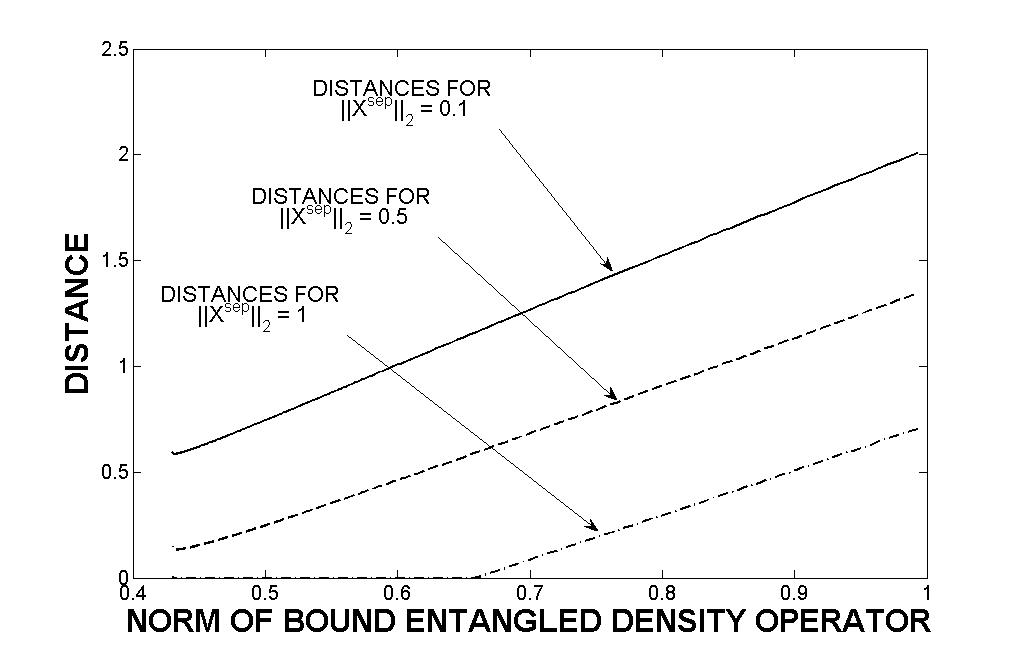}\\
  \caption{There is a linear dependence between the distance to the MVCE and the norm of the associated density operator}\label{BEdist}
\end{figure}

\part{Quantum Information Theory}

\chapter{Classical Information over Quantum Channels}

Quantum Information Theory harbors the possibility of enhancing communication with the help of an genuinely quantum resource: quantum entanglement. Entanglement thus should be regarded as new kind of information. Hence it is natural to expect that different capacities can be defined, depending on which kind of information (classical or quantum) is to be sent over a channel.

The \emph{classical capacity} $C$ is the asymptotical rate at which classical information can be transmitted through a quantum channel. Depending on whether we allow for quantum or classical coding and decoding, the classical capacity unfolds into four different capacities \cite{bennett1998qit} (see fig. \ref{fig13}). We will obviate this fact and consider only a general $C$, which will be the largest capacity among all coding and decoding possibilities. There is another classical capacity, called \emph{entanglement-assisted classical capacity} $C_{E}$. It is the rate at which classical information can be sent over a quantum channel provided that sender and receiver share an unlimited amount of entangled pairs. We will see that this capacity is larger than $C$.

There is also a \emph{quantum capacity} $Q_{1}$ for transmitting \emph{intact} quantum general states. Still there is a \emph{classically-assisted quantum capacity}, $Q_{2}$, for transmitting intact quantum states in parallel with a classical feedback channel, which permits sender and receiver to perform coordinated local operations to fight noise. The characterization of these two quantum capacities is important because it determines how much entanglement can be conveyed through a channel.

Here we will focus only on the former capacities, $C$ and $C_{E}$, since they dwell more in the spirit of classical information theory, and are also better understood. The results of quantum source coding \cite{schumacher1995qc}\cite{jozsa1994npq}, will also not be treated here for conciseness

\section{Quantum Asymptotic Equipartition Property}

In analogy with i.i.d. processes in Chapter 2, it is possible to derive typicality results for quantum  systems. Suppose that a device outputs a system in state $\varphi_{i}$ with probability $p_{i}$, with $\varphi_{i}$ necessarily orthogonal rank one density operators\footnote{This can always be done by diagonalizing the density operator}. The entropy of the system will be $S(\varrho=\sum_{i} p_{i} \varphi_{i}) = H(\mathbf{p})$. Now consider the m-fold tensor product $\varphi_{I} = \varphi_{i_{1}}\otimes\varphi_{i_{2}}\otimes\ldots\otimes\varphi_{i_{m}}$, and call it \emph{sequence density operator}. The eigenvectors of this sequence density operator live in the space $\mathcal{H}^{\otimes m} = \mathcal{H}_{1}\otimes \mathcal{H}_{2}\otimes\ldots\otimes \mathcal{H}_{m}$. Denote by $p_{I} = p_{i_{1}}p_{i_{2}}\ldots p_{i_{m}}$ the product of all probabilities corresponding to a given sequence. The sequence $\varphi_{I}$ will be typical if:

\begin{equation}
    | - \frac{1}{m}\log p_{I} - S(\varrho) | \leq \epsilon, \forall\epsilon
\end{equation}

Likewise, it is possible to define the \emph{typical subspace} of $\mathcal{H}^{\otimes m}$ as the subspace $\Lambda_{\mathcal{T}}$ spanned by the eigenvectors of all typical sequences. The orthogonal projector onto this subspace, $\Pi_{\Lambda}$, has the following properties:

\begin{equation}
    \langle\varphi_{I}, \Pi_{\Lambda}\rangle \geq 1 - \delta
\end{equation}

\begin{equation}
    \langle\varphi_{I} \Pi^{\perp}_{\Lambda}\rangle\leq\delta, \forall\delta
\end{equation}

that is, for sufficiently large m, almost all the probability is contained in the typical subspace. The dimension of the typical subspace will be bounded  by:

\begin{equation}\label{typical0}
    2^{m(S(\varrho) - \epsilon)}\leq Tr(\Pi_{\Lambda}) \leq 2^{m(S(\varrho) + \epsilon)}
\end{equation}

or equivalently,

\begin{equation}\label{QAEP}
    2^{-m(S(\varrho) + \epsilon)}\leq p_{I} \leq 2^{-m(S(\varrho) - \epsilon)}
\end{equation}

which means that as m grows, all typical sequences will tend to be equiprobable.

\subsection{Entanglement Distillation}

These typicality results will allow us to prove Theorem \ref{distillation}. For this we will state the following lemma, whose proof can be found in \cite{nielsen1999cce}:

\begin{lem} A bipartite pure state $\varphi^{AB}$ can be transformed into another pure state $\varrho^{AB}$ by LOCC if and only if the eigenvalues of their reduced density operators satisfy $\lambda(\varphi^{A})\prec\lambda(\varrho^{A})$
\end{lem}

\begin{proof}[Proof of Theorem \ref{distillation}] Suppose we have $(\varphi^{AB})^{\otimes m}$ such that its reduced density operator is $(Tr_{B}\varphi^{AB})^{\otimes m} = \sum_{I} p_{I} \varphi_{I}$. As m grows, the eigenvalues of this density operator will satisfy eq. \ref{QAEP}.

The reduced density operator of a maximally entangled state has maximum entropy $S(Tr_{B}\varrho^{AB}_{\psi_{1}}) = \log 2 = 1$. Now consider n copies of the maximally entangled state, $(\varrho^{AB}_{\psi_{1}})^{\otimes n}$. As n grows, its eigenvalues of $(Tr_{B}\varrho^{AB}_{\psi_{1}})^{\otimes n}$ will also satisfy eq.\ref{QAEP}, so they will be constrained to take values arbitrarily close to $2^{-n}$.

We have the inequality:

\begin{equation}
    2^{-m (S(\varphi^{A}) \pm \epsilon)} \leq 2^{-n \pm \epsilon'}
\end{equation}

and by previous lemma, if the entropy $S(Tr_{B}\varphi^{AB}) \approx \frac{n}{m} $, then it will be possible to transform m copies of the arbitrarily entangled state $\varphi^{AB}$ into n copies of $\varrho^{AB}_{\psi_{1}}$. This argument can be carried out symmetrically considering the other subsystem.
\end{proof}

\section{Quantum Channels}

A quantum channel $\mathcal{C}$ is a map from one algebra into another:

$$\mathcal{C}: \mathfrak{A} \longrightarrow \mathfrak{B}$$

Classically, a channel induces some \emph{noise} due to the stochastic nature of its associated transition matrix. In Quantum Mechanics time evolution of a closed system is completely deterministic. However, the system will generally couple to unaccessible degrees of freedom corresponding to dynamical variables of its environment. At the end, only the system will be observed. Tracing out the environment can introduce noise in the resulting density operator. The effect of a channel $\mathcal{C}$ onto $\varrho$ can be expressed as:

\begin{equation}
    \mathcal{C}(\varrho) = Tr_{E}(\mathbb{U}(\varrho\otimes\varrho^{E}_{1})\mathbb{U}^{\dagger})
\end{equation}

where $\varrho^{E}_{1}$ is the initial pure state of the environment, and $\mathbb{U}$ is some time evolution operator acting on the global algebra, which may, or may not, couple the system and its environment. Some physical requirements are:

\begin{itemize}
  \item It should be a \emph{completely positive} map. This stems from the fact that if $\varrho^{AB}$ is the state of a composite system, and only one of the subsystems is sent over the channel, the result still should be a density operator. This has profound consequences, such as the duality between channels and entangled states.
  \item It should be a \emph{trace preserving} map, $Tr(\mathcal{C}(\varrho)) = 1$. This is the demand that any POVM on the channel's output is normalized to one\footnote{Non-trace preserving maps are interesting to describe measurements as a channel from algebra of quantum systems to the algebra of classical systems. Non-trace preserving maps are also interesting when for some reason there is a probability leakage in the channel, i.e. when sometimes the channel produces no output at all, so $Tr(\mathcal{C}(\varrho)) \leq 1$}.
  \item It should be a \emph{convex-linear} map, $\mathcal{C}(\sum_{i}p_{i}\varrho_{i}) = \sum_{i} p_{i} \mathcal{C}(\varrho_{i}) $. In other words, channel effects should be regardless of the convex decomposition of the input density operator, pretty much in the same way that classically a transition matrix doesn't depend on the input probability distribution.
\end{itemize}

It turns out \cite{nielsen2000qca} that this requirements are necessary and sufficient to come to the \emph{operator sum} representation of channels:

\begin{equation}\label{operatorsum}
    \mathcal{C}(\varrho) = \sum^{N}_{i=1} A_{i} \varrho A^{\dagger}_{i}
\end{equation}

with $N \leq dim(\mathfrak{A}) dim(\mathfrak{B})$ and $A_{i}: \mathcal{H}_{A} \rightarrow \mathcal{H}_{B}$ are called \emph{Kraus operators}. The cannel is trace preserving, hence:

\begin{equation}
    \sum^{N}_{i=1} A^{\dagger}_{i} A_{i} = \mathbf{1}
\end{equation}

Consider an orthogonal resolution of the identity $\{\varrho^{E}_{i}=|e_{i}\rangle\langle e_{i}|\}^{N}_{i=1}$ for the environment\footnote{Here we assume implicitly that at most $d^{2}$ dimensions are necessary to model the environment for one system of dimension d} and suppose for simplicity that $\varrho = |\chi\rangle\langle\chi|$. Time evolution will correlate the state with its environment. Since the environment cannot be measured, this correlation cannot be exploited. Tracing out the environment it is easy to check that:

\begin{eqnarray}
  &&Tr_{E}(\mathbb{U}(|\chi\rangle\langle\chi|\otimes|e_{1}\rangle\langle e_{1}|)\mathbb{U}^{\dagger}) =\nonumber\\
  &=& \sum^{N}_{i=1} \langle e_{i}|[\mathbb{U}(|\chi\rangle\langle\chi|\otimes|e_{1}\rangle\langle e_{1}|)\mathbb{U}^{\dagger}]|e_{i}\rangle = \nonumber\\
  &=& \sum^{N}_{i=1} A_{i}\varrho A^{\dagger}_{i}
\end{eqnarray}

with $A_{i} = \langle e_{i}|\mathbb{U}|e_{1}\rangle$. There is a straightforward interpretation of eq. \ref{operatorsum}. From linearity and the assumption that the channel is trace preserving, $Tr(A_{i}\varrho A^{\dagger}_{i}) = p_{i}$ is the probability that the channel outputs $\varrho_{i}$. Using Bayes' rule we have that $\varrho_{i} = \frac{A_{i}\varrho A^{\dagger}_{i}}{Tr(A_{i}\varrho A^{\dagger}_{i})}$, so:

\begin{equation}
    \mathcal{C}(\varrho) = \sum_{i} p_{i} \varrho_{i}
\end{equation}

Thus, the stochastic behavior of the channel comes from the interaction of the sent system with its environment. In fact, for ideal channels, i.e. for cases where the system and its environment don't interact, evolution is unitary and only one Kraus operator is needed to define the channel. However this is not of much relevance, since the mere fact of observing a system can be described as a highly non-ideal channel(see \ref{sec:decoherence}).

\section{Classical Capacity of a Quantum Channel}

In a classical communication setting, the inference capability of the receiver is related to the mutual information between the sender and receiver's probability distributions. This capability will depend on the nature of the channel: for a noiseless channel the mutual information attains a maximum $I(\mathbf{p}^{X};\mathbf{p}^{\hat{X}}) = H(\mathbf{p}^{X})$, with $\mathbf{p}^{\hat{X}}=f(\mathbf{q}^{Y})$. For a noisy channel, mutual information can be maximized by finding the optimal probabilities of the input code.

\begin{figure}[h]
\centering
  \includegraphics[scale=0.5]{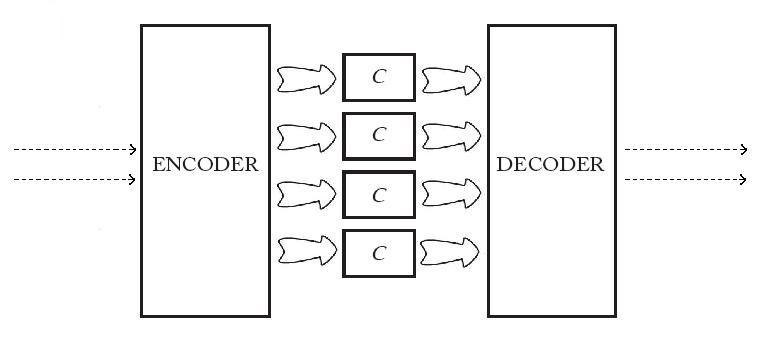}\\
  \caption{Quantum Channel. Two classical bits are encoded in four qubits and sent over the channel. Channel inputs can be entangled or not. The measurement at the decoder usually is over the joint qubit sequence.}\label{fig13}
\end{figure}

For some reason, one might want to use quantum states to encode classical bits. What makes this scenario interesting is that the quantum states states may not in general be orthogonal. In fact, this can be desirable in some cases \cite{fuchs1997nqs}\cite{peres1991odq}. Then, the Holevo's bound (see section \ref{sec:holevo}) is telling us that the maximum accessible information is bounded by:

\begin{equation}\label{holev2}
I(X:Y) \leq S(\varrho^{X}) - \sum^{n}_{i=1}p^{X}_{i}S(\varrho_{i})
\end{equation}\

This suggests that, in contrast to classical channel coding, quantum channel coding demands that two optimizations be carried out to find the optimal performance of a channel. First, the optimal measurement strategy that maximizes \ref{holev2} must be found, this is a search in the set $\mathfrak{M}$. In second place, the optimal input probability is to be found in the classical probability simplex $\mathfrak{P}_{n}$, just as in the classical case. We can now define the classical capacity of a quantum channel as:

\begin{equation}
    C= \max_{\mathbf{p}^{X}} \max_{\mathbf{M}} I(X:Y)
\end{equation}

such that $\mathbf{M}(\varrho^{X})= \mathbf{q}^{Y}$. A direct consequence of the No-Cloning Theorem (see \ref{sec:nocloning}) is that, when the system appears at the receiver's side, it must have disappeared at the sender's side. Since the random variables $X$ and $Y$ are directly related to the same system $\varrho^{X}$ at \emph{different} times, a joint probability distribution is lacking a true consistent meaning. Hence, the concept of mutual information is meaningless, as long as it refers to the information that one systems contains about itself prior to having been sent through a channel. As we will see, this subtlety precludes the use of joint typicality arguments in coding theorems of quantum channels.

The optimal measurement strategy is given by the Holevo's bound, so we have:

\begin{equation}\label{capacityquantum}
    C= \max_{\mathbf{p}^{X}} S(\varrho^{X}) - \sum^{n}_{i=1}p^{X}_{i}S(\varrho_{i})
\end{equation}

Following \cite{hausladen1996cic} \cite{holevo1998cqc} \cite{schumacher1997sci}, we state a coding theorem for quantum noisy channels:

\begin{thm}\textcolor[rgb]{0.00,0.00,1.00}{[Holevo-Schumacher-Westmoreland Theorem]} A quantum noisy channel $\mathcal{C}: \mathfrak{A} \longrightarrow \mathfrak{B}$ can be used to transmit information reliably if and only if $R\leq C$, with capacity defined as:

\begin{equation}
    C= \max_{\mathbf{p}^{X}} S(\mathcal{C}(\varrho^{X})) - \sum^{n}_{i=1}p^{X}_{i}S(\mathcal{C}(\varrho_{i}))
\end{equation}

where $\varrho_{i}$ are the input states and $\sigma_{i} = \mathcal{C}(\varrho_{i})$ represent the states to be measured by the decoder.
\end{thm}

\begin{proof}[$\Rightarrow$Proof of Achievability] Suppose that the message $w = (i_{1}, i_{2}, ..., i_{m})$ is to be sent. The sender will construct $\varrho_{w} = \varrho_{i_{1}}\otimes\varrho_{i_{2}}\otimes....\otimes\varrho_{i_{m}}$ and the channel's output will be $\sigma_{w} = \mathcal{C}^{\otimes m}(\varrho_{w})$.

The probability of successfully identifying $\sigma_{w}$ is $\langle\sigma_{w}, M_{w}\rangle$, where $M_{w}$ is a measurement for index $w$. An error will be declared with probability:

\begin{equation}\label{Perror}
    p_{e} = 2^{-mR}\sum^{2^{mR}}_{w=1} (1 - \langle\sigma_{w}, M_{w}\rangle)
\end{equation}

Classically, one could resort to joint typicality arguments to build a proof. Since quantum physics prevents us from considering the mutual information of two distributions that exist at \emph{different} times, we cannot follow this way. Instead we will consider two different applications of typicality, one concerning which sequence density operators will be more likely (much like in the classical setting), and the other sort of quantifying how many sequences can be considered to be ``close" to a \emph{fixed} sequence density operator.

Let $\bar{\sigma} = \sum^{n}_{i} p_{i} \sigma_{i}$ be an average output of the channel with spectral decomposition $\bar{\sigma} = \sum_{j} \lambda_{j}|e_{j}\rangle\langle e_{j}|$. Consider the m-fold tensor product $\bar{\sigma}^{\otimes m}=\sum_{J} \lambda_{J}|e_{J}\rangle\langle e_{J}|$, with $J=(j_{1}, j_{2}, ..., j_{m})$. Define $\Pi_{\Lambda} = \sum_{J\in\mathcal{T}} |e_{J}\rangle\langle e_{J}|$ as the projector onto the typical subspace $\Lambda_{\mathcal{T}} \in H^{\otimes m}$ spanned by the eigenvectors of all typical sequence density operators:

\begin{equation}
\mathcal{T} =\{J:  2^{-m(S(\bar{\sigma}) + \epsilon)}\leq \lambda_{J} \leq 2^{-m(S(\bar{\sigma}) - \epsilon)}\}
\end{equation}

Then:

\begin{equation}\label{typical1}
    \langle \bar{\sigma}^{\otimes m}, \Pi_{\Lambda}\rangle \geq 1 -\delta
\end{equation}

Now, let $\sigma_{w}$ be the output sequence of the channel. It has a spectral decomposition:

\begin{equation}
    \sigma_{w} = \sum_{J} \lambda^{w}_{J}|e^{w}_{J}\rangle\langle e^{w}_{J}|
\end{equation}

$\sigma_{w}$ will be the tensor product of about $mp_{1}$ copies of $\sigma_{1}$, $mp_{2}$ copies of $\sigma_{2}$, and so on... Define the \emph{average per symbol entropy} of the sequence as $\bar{S}(\sigma_{w}) = \sum_{i}p_{i}S(\sigma_{i})$. Interchanging the two definitions of entropy, we build the projector $\Pi_{w} = \sum_{J\in\mathcal{T}_{w}} |e_{J}\rangle\langle e_{J}|$, where:

\begin{equation}
\mathcal{T}_{w} =\{J:  2^{-m(\bar{S}(\sigma_{w}) + \epsilon)}\leq \lambda^{w}_{J} \leq 2^{-m(\bar{S}(\sigma_{w}) - \epsilon)}\}
\end{equation}

Then:

\begin{equation}\label{typical2}
    \langle \sigma_{w}, \Pi_{w}\rangle \geq 1 -\delta
\end{equation}

Next thing to do is to define the POVM associated to the decoder $\mathbf{M} = \{M_{w}\}^{2^{mR}}_{w=1}$. Each component of the POVM should be very close to the typical projector $\Pi_{w}$. Since only typical sequences will be assigned a codeword:

\begin{equation}
M_{w}= (\sum_{w'}\Pi_{\Lambda}\Pi_{w'}\Pi_{\Lambda})^{-\frac{1}{2}} \Pi_{\Lambda}\Pi_{w}\Pi_{\Lambda} (\sum_{w'}\Pi_{\Lambda}\Pi_{w'}\Pi_{\Lambda})^{-\frac{1}{2}}
\end{equation}

which makes sure that no non-typical sequence will be considered inside the typical subspace of the sequence $\sigma_{w}$. The generalized inverse square roots \footnote{The operator $X^{-\frac{1}{2}}$ is equal to 0 on $Ker X$ and equal to $(X^{\frac{1}{2}})^{-1}$ on $Ker X^{\perp}$} are introduced for normalization.

Now that the two concepts of typicality are introduced, return to eq. \ref{Perror}:

\begin{eqnarray}\label{Perror2}
  p_{e}&=&2^{-mR}\sum^{2^{mR}}_{w=1}(1-\sum_{J}\sum_{J'\in\mathcal{T}_{w}}\lambda^{w}_{J}\alpha^{2}_{(w,J),(w,J')}) \nonumber\\
  &\leq&2^{-mR}\sum^{2^{mR}}_{w=1}(\sum_{J\in\mathcal{T}_{w}}\lambda^{w}_{J}(1 - \alpha^{2}_{(w,J),(w,J)}) + \sum_{J\not\in\mathcal{T}_{w}}\lambda^{w}_{J})  \nonumber\\
  &\leq&2^{-mR}\sum^{2^{mR}}_{w=1}(2\sum_{J\in\mathcal{T}_{w}}\lambda^{w}_{J}(1 - \alpha_{(w,J),(w,J)}) + \sum_{J\not\in\mathcal{T}_{w}}\lambda^{w}_{J})
\end{eqnarray}

where

$$\alpha_{(w,J),(w',J')} = \langle e^{w}_{J}|\Pi_{\Lambda} (\sum_{w''}\Pi_{\Lambda}\Pi_{w''}\Pi_{\Lambda})^{-\frac{1}{2}} \Pi_{\Lambda}|e^{w'}_{J'}\rangle$$

The first inequality follows from omitting some non-positive cross terms and the relation $\sum_{J}\lambda^{w}_{J} =1$. The second inequality comes from considering the componentwise inequality $(1 - x)^{2}\leq (1 + x)(1 -x)\leq 2(1-x), x\in[0,1]$.

Once we realize that the $\alpha_{(w,J),(w',J')}$ are the entries of the square root of the Gram matrix $\Gamma = [\langle e^{w}_{J}|\Pi_{\Lambda}|e^{w'}_{J'}\rangle ] = [\gamma_{(w,J),(w',J')}] $, it is possible to express first term of the last member in eq. \ref{Perror2} as:

\begin{eqnarray}\label{Perror3}
   &&2\sum^{2^{mR}}_{w=1}\sum_{J\in\mathcal{T}_{w}}\lambda^{w}_{J}(1 - \alpha_{(w,J),(w,J)})\nonumber\\ &=&2Tr_{(w,J)}(\Lambda(\mathbb{I} - \Gamma^{\frac{1}{2}}))\nonumber\\
&=&Tr_{(w,J)}(\Lambda(\mathbb{I} - \Gamma^{\frac{1}{2}})^{2}) + Tr_{(w,J)}(\mathbb{I} - \Gamma)\nonumber\\
&\leq&Tr_{(w,J)}(\Lambda(\mathbb{I} - \Gamma)^{2}) + Tr_{(w,J)}(\mathbb{I} - \Gamma)\nonumber\\
&=& \sum^{2^{mR}}_{w=1}\sum_{J\in\mathcal{T}_{w}}\lambda^{w}_{J}[2 - 3\gamma_{(w,J),(w,J)} + \sum^{2^{mR}}_{w'=1}\sum_{J'\in\mathcal{T}_{w}}(\gamma_{(w,J),(w',J')}\gamma_{(w',J'),(w,J)})]\nonumber\\
&=&\sum^{2^{mR}}_{w=1}\sum_{J\in\mathcal{T}_{w}}\lambda^{w}_{J}[2 - 3\gamma_{(w,J),(w,J)} + \gamma^{2}_{(w,J),(w,J)} + \sum_{J'\neq J}\gamma^{2}_{(w,J),(w,J')} \nonumber\\
&&+ \sum_{w'\neq w}\sum_{J'\in\mathcal{T}_{w'}}\gamma^{2}_{(w,J),(w',J')}]
\end{eqnarray}

where $\Lambda = diag(\lambda^{w}_{J})$, and ``$Tr_{(w,J)}$" denotes the trace with respect to this joint index of the Gram matrices, instead of the usual trace over the dimension of density operators. Using the fact that $2 - 3x + x^{2}\leq 2 - 2x, x\in[0,1]$, we see that that \ref{Perror3} is upper-bounded by:

\begin{eqnarray}
  p_{e} &\leq& 2^{-mR}\sum^{2^{mR}}_{w=1}\{\sum_{J}\lambda^{w}_{J}[2 - 2\gamma_{(w,J),(w,J)}+ \sum_{J'\neq J}\gamma^{2}_{(w,J),(w,J')} \nonumber\\
&& + \sum_{w'\neq w}\sum_{J'\in\mathcal{T}_{w'}}\gamma^{2}_{(w,J),(w',J')}]+ \sum_{J\not\in\mathcal{T}_{w}}\lambda^{w}_{J}\}
\end{eqnarray}

Note that we expanded the range of the sum from $J\in\mathcal{T}_{w}$ to all $J$. Some algebra shows that it is equivalent to:

\begin{eqnarray}
    p_{e} &\leq& 2^{-mR}\sum^{2^{mR}}_{w=1}\{2Tr(\sigma_{w}(\mathbf{1} - \Pi_{\Lambda})) + Tr(\sigma_{w}(\mathbf{1} - \Pi_{\Lambda})\Pi_{w}(\mathbf{1} - \Pi_{\Lambda}))\nonumber\\
&& + \sum_{w'\neq w}Tr(\Pi_{\Lambda}\sigma_{w}\Pi_{\Lambda}\Pi_{w'}) + Tr(\sigma_{w}(\mathbf{1} - \Pi_{w}))\}\nonumber\\
&\leq&2^{-mR}\sum^{2^{mR}}_{w=1}\{3Tr(\sigma_{w}(\mathbf{1} - \Pi_{\Lambda}))+ \sum_{w'\neq w} Tr(\Pi_{\Lambda}\sigma_{w}\Pi_{\Lambda}\Pi_{w'}) + Tr(\sigma_{w}(\mathbf{1} - \Pi_{w}))\}\nonumber\\
\end{eqnarray}

here we used that $\langle e^{w}_{J}|e^{w}_{J'}\rangle = 0$ to introduce the tautology $\Pi_{\Lambda} = \mathbf{1} - \mathbf{1} + \Pi_{\Lambda}$. The last inequality follows from the fact that $Tr(\sigma_{w}(\mathbf{1} - \Pi_{\Lambda})\Pi_{w}(\mathbf{1} - \Pi_{\Lambda}))\leq Tr(\sigma_{w}(\mathbf{1} - \Pi_{\Lambda})\Pi_{w})$.

Finally, applying the concept of random coding to symmetrize over all codewords, we see that:

\begin{eqnarray}
  p_{e} &=& \sum_{\mathcal{C}}p(\mathcal{C})p_{e}(\mathcal{C}) \nonumber\\
  &\leq& \sum_{\mathcal{C}}p(\mathcal{C})\{2^{-mR} \sum^{2^{mR}}_{w=1} [3Tr(\bar{\sigma}^{\otimes m}(\mathbf{1} - \Pi_{\Lambda}))+ \nonumber\\
  &&\sum_{w'\neq w} Tr(\Pi_{\Lambda}\bar{\sigma}^{\otimes m}\Pi_{\Lambda}\Pi_{1'}) + Tr(\sigma_{1}(\mathbf{1} - \Pi_{1}))]\}\nonumber\\
  &\leq& 4\delta + (2^{mR} - 1)Tr(2^{-m(S(\bar{\sigma}) - \epsilon)}\Pi_{1'})\nonumber\\
  &\leq& 4\delta + (2^{mR} - 1)2^{-m(S(\bar{\sigma}) - \epsilon)}2^{m(\bar{S}(\sigma) + \epsilon)}
\end{eqnarray}

where we used typicality arguments of eqs. \ref{typical0}, \ref{typical1}, \ref{typical2}, and the fact that $\Pi_{\Lambda}\bar{\sigma}^{\otimes m}\Pi_{\Lambda} \leq 2^{-m(S(\bar{\sigma}) - \epsilon)}\mathbf{1}$. This proves that whenever $R\leq S(\bar{\sigma}) - \bar{S}$, the probability of error goes to zero as $m$ grows.
\end{proof}

\begin{proof}[$\Leftarrow$Weak Converse]To prove that if $R > S(\bar{\sigma}) - \bar{S}$, the error is bounded away from zero, we will use Fanno's inequality (see eq. \ref{fanno2}), as in the classical case:

\begin{equation}\label{fanno3}
H(W|Y^{m})\leq mP_{e}R + 1 = m\epsilon_{m}
\end{equation}

Technically, we are going to prove that when the error goes to zero, then the rate must necessarily be less than the capacity. Assuming that messages are equiprobable:

\begin{eqnarray}
  mR &=& H(\mathbf{p}^{W}) \nonumber\\
   &=& I(W,Y^{m}) + H(W|Y^{m})\nonumber\\
   &\leq& I(W,Y^{m}) + m\epsilon_{m}\nonumber\\
   &\leq& S(\sum^{2^{mR}}_{w=1}\frac{\sigma_{w}}{2^{mR}}) - \sum^{2^{mR}}_{w=1}\frac{S(\sigma_{w})}{2^{mR}} + m\epsilon_{m}\nonumber\\
   &\leq& \sum^{m}_{i = 1}[S(\sum^{2^{mR}}_{w=1}\frac{\sigma^{w}_{i}}{2^{mR}}) - \sum^{2^{mR}}_{w=1}\frac{S(\sigma^{w}_{i})}{2^{mR}}] + m\epsilon_{m}\nonumber\\
   &\leq& m[S(\sum^{2^{mR}}_{w=1}\frac{\sigma^{w}_{i}}{2^{mR}}) - \sum^{2^{mR}}_{w=1}\frac{S(\sigma^{w}_{i})}{2^{mR}}]  + m\epsilon_{m}\nonumber\\
   &\leq& mC + m\epsilon_{m}\nonumber
\end{eqnarray}

Second and third inequality follow from the Holevo's bound and the subadditivity of von Neumann's entropy, respectively. Last inequality is due to the definition of the capacity, since all terms in the sum are no greater than the capacity as defined in eq. \ref{capacityquantum}. Thus, we proved that if $R\geq C$ as $m$, then the error must be bounded away from zero as $m$ grows.

\end{proof}

For trace preserving channels (the ones being considered here), it was found that transmitting entangled states does not increase the capacity \cite{king2001mes}.

\section{Entanglement-Enhanced Classical Communication}

The main goal of this dissertation was to argue that entanglement can be used to increase the classical capacity of information transfer. As an example, consider the \emph{Quantum Erasure Channel} (QEC). The QEC is a map from an algebra of dimension $N$ to an algebra of dimension $N+1$. It maps an input state to itself with probability $1-\epsilon$. With probability $\epsilon$ the channel maps its input state to an \emph{erasure symbol} state, orthogonal to all input states. For the qubit case, the QEC would take as input states $\varrho_{0}=|0\rangle\langle 0|$, $\varrho_{1}=|1\rangle\langle 1|$ to $\varrho_{0}=|0\rangle\langle 0|$, $\varrho_{1}=|1\rangle\langle 1|$, and $\varrho_{2}=|2\rangle\langle 2|$, with $\langle 0,2\rangle = \langle 1,2\rangle = 0$. The classical capacity of this binary erasure channel is given by \cite{cover2006eit}:

\begin{equation}
    C = 1 - \epsilon
\end{equation}

It was already shown in \ref{sec:superdense} that sharing a maximally entangled pair permits to send two classical bits encoded in just one qubit. If the sender and receiver share an unlimited amount of maximally entangled pairs, it is possible for the sender to pre-process its entangled subsystem in such a way that the total entropy of the state will be:

\begin{equation}
    S_{T} = S(Tr_{B}\varrho^{AB}_{\psi}) + \log 2
\end{equation}

where the first term is the von Neumann's entropy of the reduced density operator, which is at a maximum for EPR pairs, and zero for separable states.  The second term is due to the choice of sender between $\varrho_{0}$ and $\varrho_{1}$\footnote{It assumes that both symbols are equiprobable, which maximizes capacity.}. It will always be larger than the entropy of any classical state. The $C_{E}$ of the QEC is given by \cite{bennett1997cqe}:

\begin{equation}
    C = 2(1 - \epsilon)
\end{equation}

\begin{figure}[h]
\centering
  \includegraphics[width=7.5cm]{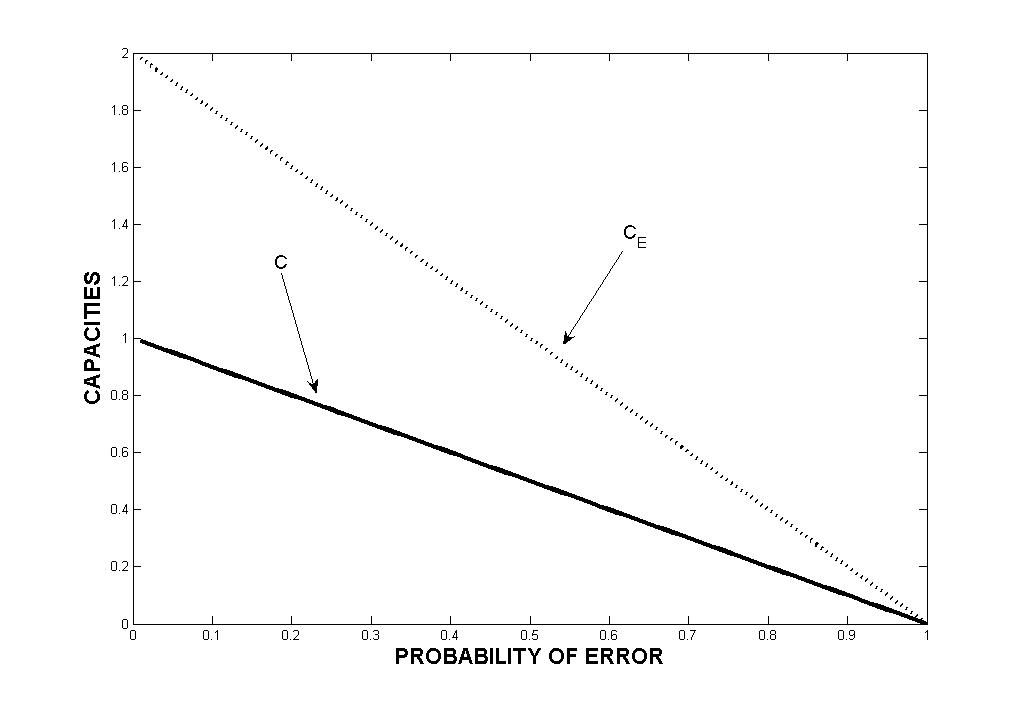}\\
  \caption{Both Capacities $C$ and $C_{E}$ versus the probability of error, $\epsilon$}\label{capacities}
\end{figure}

Now we come to the point where it is possible for us to state that \cite{bennett1999eac}:

\begin{center}
\textbf{The entanglement-assisted classical capacity will be always larger than the unassisted classical capacity.}
\end{center}

\chapter*{Conclusions and Final Words}

I feel somewhat ashamed of not having included important topics such as quantum source coding, or the duality between channels and entangled states. Moreover, at early stages of this thesis I intended to include, as well, an introduction to \emph{multiple user quantum information theory}, so at the end this work has fallen short of what it was meant to be. A decision was made on a basis of time budget, and I hope that this lack will not preclude a self-contained read of the text.

Besides the aforementioned subjects, I would find it interesting to study more in depth some other topics such as \emph{quantum rate distortion theory}, \emph{quantum signal processing} or \emph{quantum cryptography}. However this would demand too much more efforts than those expected in a master's thesis.

Remarkably, the most technical part of this work was done in the scope of \emph{convex optimization}, which is not to surprise anyone in the Information Theory community. A method for classifying entangled and separable states based on a Minimum Volume Covering Ellipsoid was devised by myself (to my knowledge, no one had done this before) for a class project. In a sense, this constitutes the \emph{state-of-the-art} part of the thesis. At the time of writing, a document explaining the method can be found in the \emph{arxiv.org} database.

A word is to be said about my previous knowledge of Quantum Information Theory, concerning the fact that I first read a paper on this subject exactly one year ago. Since then, most of my efforts have been aimed at reaching a positive semidefinite level of expertise in this field. Thus I would say that my contribution amounts to a compilation of the knowledge which I judged essential to understand the role of entanglement in classical channel capacities.

\clearpage


\begin{thebibliography}{10}

\bibitem{cover2006eit}
T.M. Cover and J.A. Thomas.
\newblock {\em {Elements of Information Theory}}.
\newblock Wiley-Interscience New York, 2006.

\bibitem{csiszar1982itc}
I.~Csiszar and J.G. Korner.
\newblock {\em {Information Theory: Coding Theorems for Discrete Memoryless
  Systems}}.
\newblock Academic Press, Inc. Orlando, FL, USA, 1982.

\bibitem{gray1990eai}
R.M. Gray.
\newblock {\em {Entropy and Information Theory}}.
\newblock Springer-Verlag New York, Inc. New York, NY, USA, 1990.

\bibitem{chiang2004gpd}
M.~Chiang and S.~Boyd.
\newblock {Geometric Programming Duals of Channel Capacity and Rate
  Distortion}.
\newblock {\em IEEE Transcations on Information Theory}, 50(2):245, 2004.

\bibitem{shannon1948mtc}
C.E. Shannon.
\newblock {A mathematical theory of communication}.
\newblock {\em Bell System Technical Journal}, 27:379--423, 1948.

\bibitem{holevo1982paq}
AS~Holevo.
\newblock {\em {Probabilistic Aspects of Quantum Theory}}.
\newblock North-Holland, Amsterdam, 1982.

\bibitem{peres1995qtc}
A.~Peres.
\newblock {\em {Quantum Theory: Concepts and Methods}}.
\newblock Kluwer Academic Pub, 1995.

\bibitem{helstrom1968sts}
C.W. Helstrom.
\newblock {\em {Statistical theory of signal detection}}.
\newblock Pergamon Press New York, 1968.

\bibitem{nielsen2000qca}
M.A. Nielsen and I.L. Chuang.
\newblock {\em {Quantum Computation and Quantum Information}}.
\newblock Cambridge University Press, 2000.

\bibitem{bratteli1996oaa}
O.~Bratteli and D.W. Robinson.
\newblock {Operator Algebras and Quantum Statistical Mechanics, Vol. 1}.
\newblock {\em NY: Springer}, 1996.

\bibitem{zeh1996ad}
HD~Zeh.
\newblock {What is achieved by decoherence?}
\newblock {\em Arxiv preprint quant-ph/9610014}, 1996.

\bibitem{joos1999eed}
E.~Joos.
\newblock {Elements of Environmental Decoherence}.
\newblock {\em Arxiv preprint quant-ph/9908008}, 1999.

\bibitem{zurek1981pbq}
WH~Zurek.
\newblock {Pointer basis of quantum apparatus: Into what mixture does the wave
  packet collapse?}
\newblock {\em Physical Review D}, 24(6):1516--1525, 1981.

\bibitem{einstein1935cqm}
A.~Einstein, B.~Podolsky, and N.~Rosen.
\newblock {Can quantum-mechanical description of physical reality be considered
  complete?}
\newblock {\em Physical Review}, 47:777, 1935.

\bibitem{bohm1952siq}
D.~Bohm.
\newblock {A suggested interpretation of the quantum theory in terms of
  'hidden'variables, I and II}.
\newblock {\em Physical Review}, 85(2):180--193, 1952.

\bibitem{bell1964epr}
JS~Bell.
\newblock {On the Einstein-Podolsky-Rosen Paradox}.
\newblock {\em Physics (US) Discontinued with Vol. 4, no. 1, 1968}, 1, 1964.

\bibitem{clauser1969pet}
J.F. Clauser, M.A. Horne, A.~Shimony, and R.A. Holt.
\newblock {Proposed Experiment to Test Local Hidden-Variable Theories}.
\newblock {\em Physical Review Letters}, 23(15):880--884, 1969.

\bibitem{aspect1982ere}
A.~Aspect, P.~Grangier, and G.~Roger.
\newblock {Experimental Realization of Einstein-Podolsky-Rosen-Bohm
  Gedankenexperiment: A New Violation of Bell's Inequalities}.
\newblock {\em Physical Review Letters}, 49(2):91--94, 1982.

\bibitem{braunstein1988itb}
S.L. Braunstein and C.M. Caves.
\newblock {Information-Theoretic Bell Inequalities}.
\newblock {\em Physical Review Letters}, 61(6):662--665, 1988.

\bibitem{horodecki2005qic}
M.~Horodecki, J.~Oppenheim, and A.~Winter.
\newblock {Quantum information can be negative}.
\newblock {\em Arxiv preprint quant-ph/0505062}, 2005.

\bibitem{devetak2006omq}
I.~Devetak and J.~Yard.
\newblock {The operational meaning of quantum conditional information}.
\newblock {\em Arxiv preprint quant-ph/0612050}, 2006.

\bibitem{schrodinger1936prb}
E.~Schrodinger.
\newblock {Probability relations between separated systems}.
\newblock {\em Proceedings of the Cambridge Philosophical Society},
  32:446--452, 1936.

\bibitem{verstraete2002seq}
F.~Verstraete.
\newblock {A Study of Entanglement in Quantum Information Theory}.
\newblock {\em These de Doctorat, Katholieke Universiteit, Leuven, Belgium},
  2002.

\bibitem{hughston1993ccq}
L.~Hughston, R.~Jozsa, and W.~Wootters.
\newblock {A Complete Classification of Quantum Ensembles Having a Given
  Density Matrix}.
\newblock 1993.

\bibitem{bennett1993tuq}
C.H. Bennett, G.~Brassard, C.~Cr{\'e}peau, R.~Jozsa, A.~Peres, and W.K.
  Wootters.
\newblock {Teleporting an unknown quantum state via dual classical and
  Einstein-Podolsky-Rosen channels}.
\newblock {\em Physical Review Letters}, 70(13):1895--1899, 1993.

\bibitem{bennett1992cvo}
C.H. Bennett and S.J. Wiesner.
\newblock {Communication via one-and two-particle operators on
  Einstein-Podolsky-Rosen states}.
\newblock {\em Physical Review Letters}, 69(20):2881--2884, 1992.

\bibitem{donald2002ute}
M.J. Donald, M.~Horodecki, and O.~Rudolph.
\newblock {The uniqueness theorem for entanglement measures}.
\newblock {\em Journal of Mathematical Physics}, 43:4252, 2002.

\bibitem{popescu1997tam}
S.~Popescu and D.~Rohrlich.
\newblock {Thermodynamics and the measure of entanglement}.
\newblock {\em Physical Review A}, 56(5):3319--3321, 1997.

\bibitem{bennett1996cpe}
C.H. Bennett, H.J. Bernstein, S.~Popescu, and B.~Schumacher.
\newblock {Concentrating partial entanglement by local operations}.
\newblock {\em Physical Review A}, 53(4):2046--2052, 1996.

\bibitem{gurvits2003cdc}
L.~Gurvits.
\newblock {Classical deterministic complexity of Edmonds' Problem and quantum
  entanglement}.
\newblock {\em Proceedings of the thirty-fifth annual ACM symposium on Theory
  of computing}, pages 10--19, 2003.

\bibitem{doherty2002dsa}
AC~Doherty, P.A. Parrilo, and F.M. Spedalieri.
\newblock {Distinguishing Separable and Entangled States}.
\newblock {\em Physical Review Letters}, 88(18):187904, 2002.

\bibitem{doherty2004cfs}
A.C. Doherty, P.A. Parrilo, and F.M. Spedalieri.
\newblock {Complete family of separability criteria}.
\newblock {\em Physical Review A}, 69(2):22308, 2004.

\bibitem{brandao2004smm}
F.G.S.L. Brand{\~a}o and R.O. Vianna.
\newblock {Separable Multipartite Mixed States: Operational Asymptotically
  Necessary and Sufficient Conditions}.
\newblock {\em Physical Review Letters}, 93(22):220503, 2004.

\bibitem{branduao2004rsp}
F.G.S.L. Brand{\u{a}}o and R.O. Vianna.
\newblock {Robust semidefinite programming approach to the separability
  problem}.
\newblock {\em Physical Review A}, 70(6):62309, 2004.

\bibitem{herreramarti2008sec}
D.A. Herrera-Mart{\'\i}.
\newblock {Scalable Ellipsoidal Classification for Entangled States}.
\newblock {\em eprint arXiv: 0806.4855}, 2008.

\bibitem{terhal2001dqe}
B.M. Terhal.
\newblock {Detecting Quantum Entanglement}.
\newblock {\em Arxiv preprint quant-ph/0101032}, 2001.

\bibitem{bruss2001ce}
D.~Bruss.
\newblock {Characterizing Entanglement}.
\newblock {\em Arxiv preprint quant-ph/0110078}, 2001.

\bibitem{peres1996scd}
A.~Peres.
\newblock {Separability Criterion for Density Matrices}.
\newblock {\em Physical Review Letters}, 77(8):1413--1415, 1996.

\bibitem{horodecki1996sms}
M.~Horodecki, P.~Horodecki, and R.~Horodecki.
\newblock {Separability of mixed states: necessary and sufficient conditions}.
\newblock {\em Physics Letters A}, 223(1-2):1--8, 1996.

\bibitem{nielsen1999cce}
MA~Nielsen.
\newblock {Conditions for a Class of Entanglement Transformations}.
\newblock {\em Physical Review Letters}, 83(2):436--439, 1999.

\bibitem{nielsen2001ssm}
MA~Nielsen and J.~Kempe.
\newblock {Separable States Are More Disordered Globally than Locally}.
\newblock {\em Physical Review Letters}, 86(22):5184--5187, 2001.

\bibitem{alt2002lfa}
H.W. Alt.
\newblock {\em {Lineare Funktionalanalysis: Eine anwendungsorientierte
  Einf{\"u}hrung}}.
\newblock Springer, 2002.

\bibitem{jamiolkowski1972ltp}
A.~Jamiolkowski.
\newblock {Linear transformations which preserve trace and positive
  semidefiniteness of operators}.
\newblock {\em Rep. Math. Phys}, 3(4):275--278, 1972.

\bibitem{horodecki1998mse}
M.~Horodecki, P.~Horodecki, and R.~Horodecki.
\newblock {Mixed-State Entanglement and Distillation: Is there a "Bound"
  Entanglement in Nature?}
\newblock {\em Physical Review Letters}, 80(24):5239--5242, 1998.

\bibitem{lewenstein2000oew}
M.~Lewenstein, B.~Kraus, JI~Cirac, and P.~Horodecki.
\newblock {Optimization of entanglement witnesses}.
\newblock {\em Physical Review A}, 62(5):52310, 2000.

\bibitem{sanpera2001snw}
A.~Sanpera, D.~Bru{\ss}, and M.~Lewenstein.
\newblock {Schmidt-number witnesses and bound entanglement}.
\newblock {\em Physical Review A}, 63(5):50301, 2001.

\bibitem{bertlmann2008gew}
R.A. Bertlmann and P.~Krammer.
\newblock {Geometric entanglement witnesses and bound entanglement}.
\newblock {\em Physical Review A}, 77(2):24303, 2008.

\bibitem{vedral1997ema}
V.~Vedral and MB~Plenio.
\newblock {Entanglement Measures and Purification Procedures}.
\newblock {\em Arxiv preprint quant-ph/9707035}, 1997.

\bibitem{wootters1998efa}
W.K. Wootters.
\newblock {Entanglement of Formation of an Arbitrary State of Two Qubits}.
\newblock {\em Physical Review Letters}, 80(10):2245--2248, 1998.

\bibitem{boyd2004co}
S.~Boyd and L.~Vandenberghe.
\newblock {\em {Convex Optimization}}.
\newblock Cambridge University Press, 2004.

\bibitem{bertlmann2002gpe}
RA~Bertlmann, H.~Narnhofer, and W.~Thirring.
\newblock {Geometric picture of entanglement and Bell inequalities}.
\newblock {\em Physical Review A}, 66(3):32319, 2002.

\bibitem{brandao2005qew}
F.G.S.L. Brand{\~a}o.
\newblock {Quantifying entanglement with witness operators}.
\newblock {\em Physical Review A}, 72(2):22310, 2005.

\bibitem{eisert2007qew}
J.~Eisert, F.~Brand{\~a}o, and KMR Audenaert.
\newblock {Quantitative entanglement witnesses}.
\newblock {\em New Journal of Physics}, 9(3):46, 2007.

\bibitem{bertlmann2005oew}
R.A. Bertlmann, K.~Durstberger, B.C. Hiesmayr, and P.~Krammer.
\newblock {Optimal entanglement witnesses for qubits and qutrits}.
\newblock {\em Physical Review A}, 72(5):52331, 2005.

\bibitem{horodecki1997sca}
P.~Horodecki.
\newblock {Separability criterion and inseparable mixed states with positive
  partial transposition}.
\newblock {\em Physics Letters A}, 232(5):333--339, 1997.

\bibitem{bennett1998qit}
CH~Bennett and PW~Shor.
\newblock {Quantum information theory}.
\newblock {\em Information Theory, IEEE Transactions on}, 44(6):2724--2742,
  1998.

\bibitem{schumacher1995qc}
B.~Schumacher.
\newblock {Quantum coding}.
\newblock {\em Physical Review A}, 51(4):2738--2747, 1995.

\bibitem{jozsa1994npq}
R.~Jozsa and B.~Schumacher.
\newblock {A new proof of the quantum noiseless coding theorem}.
\newblock {\em J Modern Optics}, pages 2343--2350, 1994.

\bibitem{fuchs1997nqs}
C.A. Fuchs.
\newblock {Nonorthogonal Quantum States Maximize Classical Information
  Capacity}.
\newblock {\em Physical Review Letters}, 79(6):1162--1165, 1997.

\bibitem{peres1991odq}
A.~Peres and W.K. Wootters.
\newblock {Optimal detection of quantum information}.
\newblock {\em Physical Review Letters}, 66(9):1119--1122, 1991.

\bibitem{hausladen1996cic}
P.~Hausladen, R.~Jozsa, B.~Schumacher, M.~Westmoreland, and W.K. Wootters.
\newblock {Classical information capacity of a quantum channel}.
\newblock {\em Physical Review A}, 54(3):1869--1876, 1996.

\bibitem{holevo1998cqc}
AS~Holevo.
\newblock {The capacity of the quantum channel with general signal states}.
\newblock {\em Information Theory, IEEE Transactions on}, 44(1):269--273, 1998.

\bibitem{schumacher1997sci}
B.~Schumacher and M.D. Westmoreland.
\newblock {Sending classical information via noisy quantum channels}.
\newblock {\em Physical Review A}, 56(1):131--138, 1997.

\bibitem{king2001mes}
C.~King and MB~Ruskai.
\newblock {Minimal entropy of states emerging from noisy quantum channels}.
\newblock {\em Information Theory, IEEE Transactions on}, 47(1):192--209, 2001.

\bibitem{bennett1997cqe}
C.H. Bennett, D.P. DiVincenzo, and J.A. Smolin.
\newblock {Capacities of Quantum Erasure Channels}.
\newblock {\em Physical Review Letters}, 78(16):3217--3220, 1997.

\bibitem{bennett1999eac}
C.H. Bennett, P.W. Shor, J.A. Smolin, and A.V. Thapliyal.
\newblock {Entanglement-Assisted Classical Capacity of Noisy Quantum Channels}.
\newblock {\em Physical Review Letters}, 83(15):3081--3084, 1999.

\end{thebibliography}
\end{document}